\documentclass[sigconf,dvipsnames]{acmart}
\usepackage{balance}
\def\BibTeX{{\rm B\kern-.05em{\sc i\kern-.025em b}\kern-.08emT\kern-.1667em\lower.7ex\hbox{E}\kern-.125emX}}

\copyrightyear{2020}
\acmYear{2020}
\setcopyright{acmcopyright}
\acmConference[WSDM '20]{The Thirteenth ACM International Conference on Web Search and Data Mining}{February 3--7, 2020}{Houston, TX, USA}
\acmBooktitle{The Thirteenth ACM International Conference on Web Search and Data Mining (WSDM '20), February 3--7, 2020, Houston, TX, USA}
\acmPrice{15.00}
\acmDOI{10.1145/3336191.3371813}
\acmISBN{978-1-4503-6822-3/20/02}

\usepackage{bm,bbm}
\usepackage{amsmath,amsthm,amssymb}
\usepackage{xspace}
\usepackage{float}								% More float options for illustrations
\usepackage{tikz}
\usetikzlibrary{positioning,calc}
\usepackage{xcolor}
\definecolor{cRed}{HTML}{DA5527}
\definecolor{cGold}{HTML}{EEB11D}
\definecolor{cBlue}{HTML}{0A73B9}
\definecolor{cGreen}{HTML}{008D0A}
\definecolor{cPurple}{HTML}{990099}

\usepackage{subfig}
\usepackage{graphicx}
\usepackage[algo2e,ruled,noend,linesnumbered]{algorithm2e}
\usepackage{booktabs}							% enhance the quality of tables
\usepackage{array}								% customize margins in tables
\usepackage{bbding} 							% for the \CheckmarkBold symbol
\usepackage{multicol}
\usepackage{enumitem}
\setlist{noitemsep, leftmargin=15pt}

\allowdisplaybreaks

\theoremstyle{definition}
\newtheorem{defn}{Definition}
\theoremstyle{acmdefinition}	% use acm style instead

\newtheorem*{rem*}{Remark}

\newcommand{\hide}[1]{}
\def\a{\mathbf{a}}

\def\f{\mathbf{f}}
\def\g{\mathbf{g}}
\def\h{\mathbf{h}}

\def\v{\mathbf{v}}
\def\w{\mathbf{w}}
\def\x{\mathbf{x}}

\def\I{\mathbf{I}}

\def\W{\mathbf{W}}
\def\X{\mathbf{X}}

\def\1{\mathbf{1}}
\def\0{\mathbf{0}}

\def\cC{\mathcal{C}}

\def\cL{\mathcal{L}}
\def\cM{\mathcal{M}}

\def\cO{\mathcal{O}}

\def\cS{\mathcal{S}}
\def\cT{\mathcal{T}}

\def\cV{\mathcal{V}}

\def\bbR{\mathbb{R}}

\def\bbI{\mathbbm{I}}

\def\defas{\triangleq}
\def\Pr{P}

\def\method{\texttt{Auto\-Block}\xspace} % \- is a discretionary hyphens that only be displayed when it's necessary.
\def\Movie{\texttt{Movie}\xspace}
\def\Music{\texttt{Music}\xspace}
\def\Grocery{\texttt{Grocery}\xspace}
\def\automation{\textbf{Automation}\xspace}
\def\scalability{\textbf{Scalability}\xspace}
\def\effectiveness{\textbf{Effectiveness}\xspace}

\def\SeqEnc{\mathrm{SeqEnc}}

\def\similarity{\sigma}

\newcommand{\myparagraph}[1]{\smallskip \noindent \textbf{#1}\hspace{0.5em}}
\newcommand{\reducemargin}{\vspace{-0.15in}}

\begin{document}

\title{AutoBlock: A Hands-off Blocking Framework for Entity Matching}

\author{Wei Zhang}
\authornote{Work performed during internship at Amazon.}
\email{zhangwei@cs.wisc.edu}
\affiliation{
  \institution{University of Wisconsin-Madison}
}

\author{Hao Wei}
\email{wehao@amazon.com}
\affiliation{
  \institution{Amazon.com}
}

\author{Bunyamin Sisman}
\email{bunyamis@amazon.com}
\affiliation{\institution{
Amazon.com}}

\author{Xin Luna Dong}
\email{lunadong@amazon.com}
\affiliation{
  \institution{Amazon.com}
}

\author{Christos Faloutsos}
\email{christos@cs.cmu.edu}
\affiliation{
  \institution{Carnegie Mellon University}
}

\author{David Page}
\email{david.page@duke.edu}
\affiliation{
  \institution{Duke University}
}

\renewcommand{\shortauthors}{Wei Zhang et al.}

\begin{abstract}
Entity matching seeks to identify data records over one or multiple data sources that refer to the same real-world entity.
Virtually every entity matching task on large datasets requires  blocking, a step that reduces the number of record pairs to be matched.
However, most of the traditional blocking methods are learning-free and key-based, and their successes are largely built on laborious human effort in cleaning data and designing blocking keys.

In this paper, we propose \method, a novel hands-off blocking framework for entity matching, based on similarity-preserving representation learning and nearest neighbor search.
Our contributions include:
(a) \textbf{Automation}: \method frees users from  laborious data cleaning and blocking key tuning.
(b) \textbf{Scalability}: \method has a sub-quadratic total time complexity and can be easily deployed for millions of records.
(c) \textbf{Effectiveness}: \method outperforms a wide range of competitive baselines on multiple large-scale, real-world datasets, especially when datasets are dirty and/or unstructured.
%, even if only minimum human effort is involved.
\end{abstract}

%   <ccs2012>
%     <concept>
%       <concept_id>10002951.10002952.10003219.10003223</concept_id>
%       <concept_desc>Information systems~Entity resolution</concept_desc>
%       <concept_significance>500</concept_significance>
%     </concept>
%     <concept>
%       <concept_id>10002951.10002952.10003219.10003183</concept_id>
%       <concept_desc>Information systems~Deduplication</concept_desc>
%       <concept_significance>500</concept_significance>
%     </concept>
%     <concept>
%       <concept_id>10003752.10010070.10010111.10011733</concept_id>
%       <concept_desc>Theory of computation~Data integration</concept_desc>
%       <concept_significance>500</concept_significance>
%     </concept>
%   </ccs2012>

% \ccsdesc[500]{Information systems~Entity resolution}
% \ccsdesc[500]{Information systems~Deduplication}
% \ccsdesc[500]{Theory of computation~Data integration}


\maketitle

\section{Introduction}
\label{sec:intro}
Entity matching seeks to identify data records over one or multiple data sources that refer to the same real-world entities.
In the era of Big Data and data science, entity matching is playing an increasingly critical role as the value of the data expands exponentially when they are linked to other data to create a unified repository~\cite{Dong2013}.
An exhaustive pairwise comparison grows quadratically with the number of records, which is unaffordable for datasets of even moderate size.
As a result, virtually every entity matching task on large datasets requires \emph{blocking}, a step that effectively reduces the number of record pairs to be considered for matching without potentially ruling out true matches.

A successful application of blocking to an entity matching task should fulfill the following four desiderata:
First, blocking, ideally, should not leave out any true matches (i.e., high \emph{recall}), since only the candidate record pairs generated by blocking will be further examined in the downstream matching step.
Second, the number of candidate pairs should be small so that the cost of applying a usually computationally-expensive matching algorithm is controlled. Therefore, it is desired to have a small ratio of the number of candidate pairs to the number of entities (Pair-Entity ratio, or \emph{P/E ratio}).
Third, \emph{human effort} should not be overspent during the whole blocking process; man-hours on cleaning data and tuning the configuration for blocking algorithms need to be minimized.
Last but not least, the blocking algorithm should be \emph{scalable} enough to handle millions of records.

Although the problem of blocking has been studied for decades, to the best of our knowledge, the dominant and most widely used methods in practice are  key-based methods.
The main idea of these methods is to divide records into a collection of blocks based on several human-crafted \emph{blocking keys} such that we only perform comparisons only among records co-occurring in the same blocks.
To improve recall, many efforts have been focusing on generating multiple customized blocking key~\cite{Gravano2003,Aizawa2005}  on individual attributes or aggregated attributes~\cite{Papadakis2013}.

\myparagraph{Challenges}
The foremost challenge for blocking is the \emph{unnormalization}, or namely the various types of noise, prevalent in the real-world data. As an illustrative example, consider two matched record pairs in Table~\ref{tab:noisy-pair} for songs.
While each tuple in the pair resembles the other, a few common cases of unnormalization can still be observed:
(a) ``Blowin'\hspace{.1em}'' is \emph{misspelled} as ``Blowing'';
(b) \emph{missing values} appear on many attributes;
(c) ``Michael Bubl\'e'' is moved from Composer to Song Writer, which may have resulted from the \emph{ambiguity} in schema definition;
(d) the title of Record 2 contains an extra version description ``[remix]'', possibly due to \emph{imperfect extraction}.

\begin{table}[!tbp]
\centering
\caption{An example of two matched pairs. Various cases of unnormalization are observed.}
\label{tab:noisy-pair}
\resizebox{\columnwidth}{!}{\begin{tabular}{
	>{\hspace{-1mm}} l <{\hspace{-1mm}}
    >{\hspace{-1mm}} l <{\hspace{-1mm}}
    >{\hspace{-1mm}} l <{\hspace{-1mm}}
    >{\hspace{-1mm}} l <{\hspace{-1mm}}
    >{\hspace{-1mm}} l <{\hspace{-1mm}}}
	\toprule
	\textbf{ID} & \textbf{Title}            & \textbf{Album}             & \textbf{Composer} & \textbf{Song Writer} \\ \midrule
	1           & Me and Mrs.\ Jones         & Call Me Irresponsible      & Michael Bubl\'e   &                      \\
	2           & Me and Mrs.\ Jones [remix] &                            &                   & Michael Bubl\'e      \\ \midrule
	3           & Blowin' in the Wind       & The Freewheelin' Bob Dylan & Bob Dylan         &                      \\
	4           & Blowing in the Wind       &                            & Bob Dylan         &                      \\ \bottomrule
\end{tabular}}

\end{table}

The result of the prevalence of unnormalization in real-world data is that blocking becomes rather challenging with traditional key-based blocking methods. This is because these methods rely on exact matching of blocking keys; thus, to deal with the unnormalized data one would have to carefully choose among a large number of combinations of different data cleaning strategies and blocking key design~\cite{Christen2012,Papadakis2016,Ebraheem2018}.
It is often the case that these decisions are \emph{dataset-specific} and not obvious even to domain experts, and many iterations of trial-and-error have to be implemented~\cite{Doan2017}.

As a concrete example, let us consider a typical blocking process for the song records in Table~\ref{tab:noisy-pair}.
A user may start with Title as a blocking key, then realize it covers few true matches because of the prevalence of the unnormalized texts in titles.
Next, the user may try various ways to clean the titles (such as removing all punctuation and version descriptions) and generate multiple customized blocking keys (such as using prefixes, suffixes and/or character/token n-grams of the titles). These attempts, however, need to be made incrementally, and usually cannot be applied altogether, since combining all of them often becomes overkill and results in an unaffordably large P/E ratio.
Furthermore, the user typically has to replicate all these efforts with the other attributes, as Title alone cannot produce high enough recall. Yet other attributes may have their own issues, such as low coverage and extremely large frequencies of particular attribute values (e.g., ``Bob Dylan'' in Composer).
Even worse, the user may need to manually recognize the correlation among a set of attributes, and create an aggregated attribute to assist the blocking (e.g., combining Composer and Song Writer).

Some non-key-based blocking methods (such as MinHash blocking~\cite{Liang2014}) can partially handle the unnormalization issue by supporting fuzzy-matching on attribute values. But these methods rely purely on lexical evidence, so they can still fall short on recall, or obtain reasonable recall but sacrifice P/E ratio, thus leading to a high comparison cost in the downstream matching step.

Therefore, the process of blocking on large-scale, unnormalized real-world data can be costly in human labor; even a well-educated domain expert often needs to spend \emph{days or  weeks} in order to achieve satisfactory blocking results.

\myparagraph{Our Solution}
In this paper, we seek to build a general blocking approach that achieves high recall, low P/E ratio, scalability, and minimum human effort, simultaneously.

We begin with an intuition as follows: \emph{if we had a good similarity metric $\similarity(\cdot, \cdot)$ for quantifying the similarity of any record pair, and could afford to apply the metric $\similarity$ to all possible pairs in the data source, blocking would be done by simply retrieving the nearest neighbors (NNs) for each record}. However, substantiating this scheme is rather challenging, since a good metric $\similarity$ for blocking is usually unknown a priori, and finding NNs is inefficient for most non-trivial $\similarity$'s. Two design questions thus arise naturally:
\begin{itemize}[leftmargin=8mm]
\item [\textbf{(Q1)}] \textbf{Similarity Metric}: \emph{Is it possible to automatically identify a good similarity metric for blocking}?
\item [\textbf{(Q2)}]\textbf{Fast NN Search}: \emph{Given the identified, potentially non-trivial similarity metric, can we find nearest neighbors for each record efficiently}?
\end{itemize}

%a set of pairwise labels that indicate which record pairs are matched has to be collected either by manual labeling or from some strong keys. Yet these labels are largely unused in existing blocking methods.

To answer these two questions, we propose \method, a hands-off blocking framework on tabular records (tuples). To automatically identify a good similarity metric, \method  utilizes a set of pairwise labels that indicates which record pairs are matched, and learns a neural network architecture that produces, for similar tuple pairs, similar real-valued representations (named \emph{signatures}), measured under some standard metric.
Thus, a similarity metric $\similarity$ for tuples is implicitly learned as the composition of the signature function (the neural network architecture) and the standard similarity metric for signatures.
To further enable efficient approximate NN search, we choose the metric for signatures to be cosine and apply cross-polytope locality-sensitive hashing (LSH)~\cite{Andoni2015a}---a theoretically optimal LSH family for cosine similarity---to retrieve the NNs for each tuple in sublinear time.

% \begin{figure}[tbp]
% \centering
% % \includegraphics[width=0.7\columnwidth]{figures/performance_scatterplot}
% \resizebox{0.8\columnwidth}{!}{\input{figures/performance_scatterplot}}
% \caption{Our method \method achieves better blocking results with much less human effort compared to alternative methods blocking.}
% \label{fig:performance-curve}
% \end{figure}

\myparagraph{Contributions}
We now underscore our main contributions:
\begin{itemize}
\item \automation: We propose a novel hands-off blocking framework, \method, that frees users from the tedious and laborious processes of data cleaning, and designing and tuning blocking keys.
% \item We propose a sequential training algorithm to learn a model that outputs uncorrelated signatures. This uncorrelation property improves the recall, and
\item \scalability: We show that \method has a sub-quadratic total time complexity for generating the candidate pairs for all tuples and thus can be easily deployed for millions of tuples.

\item \effectiveness: We evaluate \method on multiple large-scale, real-world datasets of various domains, and show that our method outperforms a wide range of competitive baselines on dirty and unstructured datasets, with minimum human effort involved.
% Figure~\ref{fig:performance-curve} illustrates the advantages of \method compared to other blocking methods.
% \footnote{The positions for various methods are based on our experimental results in Section~\ref{subsec:effectiveness}, as well as a well-received comparative study on blocking methods~\cite{Papadakis2016}.}

\end{itemize}

The rest of the paper is organized as follows. We start with notation and problem definition in Section~\ref{sec:background}. Then we further elaborate our intuition---blocking as NN search---and give an overview of the architecture of \method in Section~\ref{sec:solution_overview}. We formally present the five major steps of \method in Section~\ref{sec:method}. Section~\ref{sec:experiments} shows our experimental results, and Section~\ref{sec:related-work} discusses the related work. Finally, we conclude and list several future directions in Section~\ref{sec:conclusion}.

\section{Problem Definition}
\label{sec:background}
Suppose our dataset consists of $n$ tuples, and each tuple has $m$ attributes. We denote the $i$-th tuple by $\x_i \defas [\a_{i1}, \a_{i2}, \ldots, \a_{im}]$, where $\a_{ij} $ is the $j$-th attribute value for $\x_i$ and $\defas$ stands for ``is defined as.'' We use $[n]$ as a shorthand for the set $\{1, 2, \ldots, n\}$.
%For the purpose of presentation, we assume without loss of generality that all attributes are textual and atomic; nevertheless, our method can easily generalize to other non-textual, non-atomic attributes.
% For non-textual attributes such as numbers or dates, we convert them into their text representations; for non-atomic values (e.g., sets), we concatenate all their components.
In this way, each attribute value $\a_{ij}$ can be represented as a sequence of tokens ${[w_{ijk}]}_{k = 1}^{l_{ij}}$, where $l_{ij}$ is the sequence length for $\a_{ij}$ and $w_{ijk}$ is the $k$-th token. We assume that all tokens are drawn from a unified vocabulary $\cV$. We emphasize two important properties of the vocabulary $\cV$ for unnormalized text:
(a) openness---$\cV$ can contain out-of-vocabulary tokens and have infinite cardinality; and (b) prevalence of missing values---many $l_{ij}$'s can be zero.

We now give a formal definition of blocking as follows.
\begin{defn}[Blocking]
Given a data source $\X \defas [\x_1, \ldots, \x_n]$ of $n$ tuples, \emph{blocking} outputs a subset of candidate pairs $\cC \subseteq [n] \times [n]$, such that for any $(i, i') \in \cC$, tuple $x_i$ and tuple $x_{i'}$ are likely to refer to same entity.
% \added[remark={dislike, but don't how to write it better}]{such that $(\x_i, \x_{i'})$ is likely to be a true pair for each $(i, i') \in \cC$}.
\end{defn}
\begin{rem*}[]
As noted in the definition, high recall is the foremost requirement for blocking;  nevertheless, it is also important to control the size of $\cC$ to achieve the goal of prescreening for matching.
\end{rem*}

In addition to the tuple set $\X$, we also assume that we are able to access a \emph{positive label set} $\cL \subseteq [n] \times [n]$, such that $(\x_i, \x_{i'})$ is a match for all $(i, i') \in \cL$.
These positive labels can be generated with certain strong keys whenever available (UPC code for grocery products, ISBN numbers for books, SSN for residents, etc.), or obtained by manual annotation.
Note that this $\cL$ can be reused for training the downstream matching algorithm, which requires collecting positive labels anyway; thus requiring such $\cL$ in blocking does not incur additional human effort.

\section{Beyond Blocking: Nearest Neighbor Search}
\label{sec:solution_overview}

We hypothesize that if there is a perfect similarity metric, and efficiency is not a concern, blocking can be achieved by NN search. We refer to this scheme as \emph{NN blocking} and specify it in Algorithm~\ref{alg:nn_blocking}.

\begin{algorithm2e}[htbp]
\caption{Nearest neighbor (NN) blocking}
\label{alg:nn_blocking}

\KwIn{tuple set $\X$, a similarity  metric $\similarity(\cdot, \cdot)$, and threshold $\theta$}
\KwOut{candidate pairs $\cC$}
$\cC := \emptyset$ \;
\For{$i = 1, \ldots, n$}{
	$\cC := \cC \cup \left\{(i, i')|\, \similarity(\x_i, \x_{i'}) > \theta, \forall i' < i  \right\}$ \;
}

\end{algorithm2e}

In fact, a wide range of existing blocking methods can be viewed as special cases of NN blocking, with their own similarity metrics.
Behind the traditional key-based blocking methods, for example, are binary similarity metrics (row 2--5 in Table~\ref{tab:key-and-similarity-metric}).
This observation exposes the key reason why these methods are susceptible to unnormalized data: their similarity metrics rely on exact string matching and are too coarse-grained.

Another example is MinHash blocking, based on set-based similarity (row 6 in Table~\ref{tab:key-and-similarity-metric}). MinHash blocking first converts each tuple into a set of representative pieces that typically comprise individual tokens and token n-grams, then measures the similarity of tuples based on their set representations. Jaccard similarity and its LSH family---MinHash---are used to provide efficient approximate NN search.
However, Jaccard similarity only captures the lexical similarity between tuples and thus can be suboptimal for difficult domains where syntactic or semantic similarity is required.

\begin{table}[tbp]
\centering
\caption{Common blocking methods and their corresponding similarity metrics}
\resizebox{\columnwidth}{!}{\begin{tabular}{lll}
\toprule
\textbf{Method} 		& \textbf{Key (Example)} & \textbf{Similarity Metric} \\
\midrule
Single Key    & $f(x)=x.\text{title}$  & $\bbI(f(x) = f(y))$   \\
\midrule
Conjunctive Key~\cite{Bilenko2006,Michelson2006}
	& $f(x) = (x.\text{title}, x.\text{album})$
    & $\bbI(f(x) = f(y))$ \\
\midrule
Disjunctive Key~\cite{Bilenko2006,Michelson2006}
	& $\begin{aligned} & f_1(x) = x.\text{title} \\
    				   & f_2(x) = x.\text{album} \\ \end{aligned}$
    & $\begin{aligned} & \bbI(f_1(x) = f_1(y)) \\
       				   & \vee \bbI(f_2(x) = f_2(y)) \\
       \end{aligned}$ \\
\midrule
Customized Key~\cite{Gravano2003,Aizawa2005,Papadakis2013}
	& $\begin{aligned} f(x) = \text{firstTwoToken}(x.\text{title}) \end{aligned}$
    & $\bbI(f(x) = f(y))$ \\
\midrule
MinHash~\cite{Liang2014}
	& $\begin{aligned} f(x) = \text{nGrams}(x.\text{title}) \end{aligned}$
    & $\begin{aligned}
    	\text{JaccardSim}(f(x), f(y))
    \end{aligned}$ \\
\bottomrule
\end{tabular}
}
\label{tab:key-and-similarity-metric}
\end{table}

\myparagraph{Overview of Our Architecture}
Our framework \method follows the scheme of NN blocking and leverages a positive label set to implicitly learn the similarity metric. The overall architecture of \method comprises five steps, as illustrated in Figure~\ref{fig:overall-architecture}, in which \emph{steps (1)--(4) together form our solution for design question \textbf{Q1} and step (5) is our solution for \textbf{Q2}}.
We briefly describe the five steps below and explain them in details in the next section.
\begin{enumerate}
\item \textbf{Token embedding}: A word-embedding model transforms each token to a token embedding  (Section~\ref{subsec:token-embedding}).

\item \textbf{Attribute embedding}: For each attribute value of a tuple, an attention-based neural network encoder converts the input sequence of token embeddings to an attribute embedding (Sectio~\ref{subsec:attribute-embedding}).

\item \textbf{Tuple signature}: Multiple signature functions combine the attribute embeddings of each tuple and produce multiple tuple signatures (one per signature function) (Sectio~\ref{subsec:tuple-signature}).

\item \textbf{Model training}: Equipped with the positive label set, the model is trained with an objective that maximizes the differences of the cosine similarities between the tuple signatures of matched pairs and between unmatched pairs (Section~\ref{subsec:supervised-training}).

\item \textbf{Fast NN search}: The learned model is applied to compute the signatures for all tuples, and an LSH family for cosine similarity is used to retrieve the nearest neighbors for each tuple to generate candidate pairs for blocking (Section~\ref{subsec:LSH}).
\end{enumerate}

\begin{figure}[tbp]
\centering
\resizebox{1.0\columnwidth}{!}{\begin{tikzpicture}[
	axis/.style={line width=3, -latex, color=cBlue!80!white},
	embedding/.style={draw, thick, minimum width=1*2ex, minimum height=6*2ex, inner sep=0},
	label/.style={node font=\bf \Large},
	square/.style={rectangle, minimum width=2ex, minimum height=2ex, inner sep=0}
	]

  \node (T1) at (0,0) {
    \begin{tikzpicture}
		\coordinate[square, fill=cBlue!70!white, draw] (e00) at (0, 0);
		\coordinate[square, fill=cBlue!70!white, draw, right=0 of e00] (e01);
		\coordinate[square, fill=cBlue!55!white, draw, right=0 of e01] (e02);
		\coordinate[square, fill=cBlue!85!white, draw, above=0 of e00] (e10);
		\coordinate[square, fill=cBlue!40!white, draw, right=0 of e10] (e11);
		\coordinate[square, fill=cBlue!50!white, draw, right=0 of e11] (e12);
        \coordinate[square, fill=cBlue!30!white, draw, above=0 of e10] (e20);
		\coordinate[square, fill=cBlue!60!white, draw, right=0 of e20] (e21);
		\coordinate[square, fill=cBlue!70!white, draw, right=0 of e21] (e22);
        \coordinate[square, fill=cBlue!80!white, draw, above=0 of e20] (e30);
        \coordinate[square, fill=cBlue!40!white, draw, right=0 of e30] (e31);
        \coordinate[square, fill=cBlue!60!white, draw, right=0 of e31] (e32);
        \coordinate[square, fill=cBlue!50!white, draw, above=0 of e30] (e40);
        \coordinate[square, fill=cBlue!30!white, draw, right=0 of e40] (e41);
        \coordinate[square, fill=cBlue!70!white, draw, right=0 of e41] (e42);
	\end{tikzpicture}
	};

  \node[above=2 of T1.west, anchor=west]  (T2) {
	\begin{tikzpicture}
		\coordinate[square, fill=cGreen!20!white, draw] (e00) at (0, 0);
		\coordinate[square, fill=cGreen!40!white, draw, right=0 of e00] (e01);
		\coordinate[square, fill=cGreen!60!white, draw, right=0 of e01] (e02);
		\coordinate[square, fill=cGreen!60!white, draw, right=0 of e02] (e03);
		\coordinate[square, fill=cGreen!100!white, draw, above=0 of e00] (e10);
		\coordinate[square, fill=cGreen!40!white, draw, right=0 of e10] (e11);
		\coordinate[square, fill=cGreen!60!white, draw, right=0 of e11] (e12);
        \coordinate[square, fill=cGreen!100!white, draw, right=0 of e12] (e13);
        \coordinate[square, fill=cGreen!60!white, draw, above=0 of e10] (e20);
		\coordinate[square, fill=cGreen!60!white, draw, right=0 of e20] (e21);
		\coordinate[square, fill=cGreen!80!white, draw, right=0 of e21] (e22);
        \coordinate[square, fill=cGreen!40!white, draw, right=0 of e22] (e23);
        \coordinate[square, fill=cGreen!60!white, draw, above=0 of e20] (e30);
        \coordinate[square, fill=cGreen!100!white, draw, right=0 of e30] (e31);
        \coordinate[square, fill=cGreen!40!white, draw, right=0 of e31] (e32);
        \coordinate[square, fill=cGreen!60!white, draw, right=0 of e32] (e33);
        \coordinate[square, fill=cGreen!80!white, draw, above=0 of e30] (e40);
        \coordinate[square, fill=cGreen!80!white, draw, right=0 of e40] (e41);
        \coordinate[square, fill=cGreen!60!white, draw, right=0 of e41] (e42);
        \coordinate[square, fill=cGreen!20!white, draw, right=0 of e42] (e43);
	\end{tikzpicture}
  };
  \node[above=2 of T2.west, anchor=west]  (T3) {
	\begin{tikzpicture}
		\coordinate[square, fill=cPurple!20!white, draw] (e00) at (0, 0);
		\coordinate[square, fill=cPurple!40!white, draw, right=0 of e00] (e01);
		\coordinate[square, fill=cPurple!60!white, draw, above=0 of e00] (e10);
		\coordinate[square, fill=cPurple!20!white, draw, right=0 of e10] (e11);
        \coordinate[square, fill=cPurple!60!white, draw, above=0 of e10] (e20);
		\coordinate[square, fill=cPurple!40!white, draw, right=0 of e20] (e21);
        \coordinate[square, fill=cPurple!60!white, draw, above=0 of e20] (e30);
        \coordinate[square, fill=cPurple!20!white, draw, right=0 of e30] (e31);
        \coordinate[square, fill=cPurple!80!white, draw, above=0 of e30] (e40);
        \coordinate[square, fill=cPurple!40!white, draw, right=0 of e40] (e41);
	\end{tikzpicture}
  };

  \node[label, minimum width=60, left=0 of T1] (L1) {Attribute m};
  \node[label, minimum width=60, left=0 of T2] (L2) {Attribute 2};
  \node[label, minimum width=60, left=0 of T3] (L3) {Attribute 1};

  \node[right=2 of T1]  (A1) {
    \begin{tikzpicture}
		\coordinate[square, fill=cBlue!20!white, draw] (e0) at (0, 0);
		\coordinate[square, fill=cBlue!60!white, draw, above=0 of e0] (e1);
		\coordinate[square, fill=cBlue!40!white, draw, above=0 of e1] (e2);
		\coordinate[square, fill=cBlue!80!white, draw, above=0 of e2] (e3);
		\coordinate[square, fill=cBlue!100!white, draw, above=0 of e3] (e4);
	\end{tikzpicture}
  };
  \node[above=2 of A1.west, anchor=west]  (A2) {
	\begin{tikzpicture}
		\coordinate[square, fill=cGreen!60!white, draw] (e0) at (0, 0);
		\coordinate[square, fill=cGreen!40!white, draw, above=0 of e0] (e1);
		\coordinate[square, fill=cGreen!80!white, draw, above=0 of e1] (e2);
		\coordinate[square, fill=cGreen!20!white, draw, above=0 of e2] (e3);
		\coordinate[square, fill=cGreen!80!white, draw, above=0 of e3] (e4);
	\end{tikzpicture}
  };
  \node[above=2 of A2.west, anchor=west]  (A3) {
	\begin{tikzpicture}
		\coordinate[square, fill=cPurple!40!white, draw] (e0) at (0, 0);
		\coordinate[square, fill=cPurple!80!white, draw, above=0 of e0] (e1);
		\coordinate[square, fill=cPurple!20!white, draw, above=0 of e1] (e2);
		\coordinate[square, fill=cPurple!60!white, draw, above=0 of e2] (e3);
		\coordinate[square, fill=cPurple!20!white, draw, above=0 of e3] (e4);
	\end{tikzpicture}
  };

  \node[right=2.2 of $(A1)!0.5!(A2)$]  (S1) {
	\begin{tikzpicture}
		\coordinate[square, fill=cRed!80!white, draw] (e0) at (0, 0);
		\coordinate[square, fill=cRed!40!white, draw, above=0 of e0] (e1);
		\coordinate[square, fill=cRed!80!white, draw, above=0 of e1] (e2);
		\coordinate[square, fill=cRed!20!white, draw, above=0 of e2] (e3);
		\coordinate[square, fill=cRed!60!white, draw, above=0 of e3] (e4);
	\end{tikzpicture}
  };
  \node[right=2.2 of $(A2)!0.5!(A3)$]  (S2) {
	\begin{tikzpicture}
		\coordinate[square, fill=cGold!100!white, draw] (e0) at (0, 0);
		\coordinate[square, fill=cGold!40!white, draw, above=0 of e0] (e1);
		\coordinate[square, fill=cGold!60!white, draw, above=0 of e1] (e2);
		\coordinate[square, fill=cGold!40!white, draw, above=0 of e2] (e3);
		\coordinate[square, fill=cGold!80!white, draw, above=0 of e3] (e4);
	\end{tikzpicture}
  };

  \node[label, below=1.2 of T1] {1. Token Embedding};
  \node[label, below=0.5 of A1] {2. Attribute Embedding};
  \node[label, below=0.8 of S1] {3. Tuple Signature};
  \node[label, right=0.5 of S2] (step4) {4. Model Training};
  \node[label, right=3 of $(S1)!0.5!(S2)$] (step5) {5. LSH-Based NN Search};
  \node[label, below=0 of step5] {\input{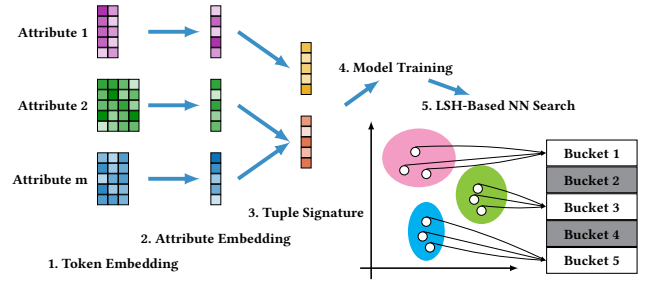}};

  %\draw[axis] (L1)-- (T1);
  %\draw[axis] (L2)-- (T2);
  %\draw[axis] (L3)-- (T3);
  \draw[axis] ([xshift=-45]A1.west)-- ([xshift=-5]A1.west);
  \draw[axis] ([xshift=-45]A2.west)-- ([xshift=-5]A2.west);
  \draw[axis] ([xshift=-45]A3.west)-- ([xshift=-5]A3.west);
  \draw[axis] ([xshift=10]A1.east)-- ([xshift=-5]S1.west);
  \draw[axis] ([xshift=10]A2.east)-- ([xshift=-5]S1.west);
  \draw[axis] ([xshift=10]A3.east)-- ([xshift=-5]S2.west);
  \draw[axis] ([xshift=30]$(S1)!0.5!(S2)$)-- (step4);
  \draw[axis] (step4)-- (step5);

\end{tikzpicture}}
\caption{Overall architecture of \method}
\label{fig:overall-architecture}
\end{figure}

\section{Proposed: \method}
\label{sec:method}
In this section, we present the five steps of \method in details.

\subsection{Token Embedding}
\label{subsec:token-embedding}

% \comment[id=CF]{drop this subsection - it has no novelty - it should
% not draw attention}
% \comment[id=CF]{drop}
% \comment{drop this}

The first step of \method is to convert each token into a low-dimensional embedding vector using a word embedding model.
%In recent years, given the ability to produce semantic-preserving embeddings for words, word embedding models have been widely applied and have led to performance boosts in myriad NLP tasks, especially when the embedding models are integrated into deep neural network based models.
%Exemplary word embedding models include word2vec~\cite{Mikolov2013}, GloVe~\cite{Gokhale2014} and fastText~\cite{Bojanowski2017}.
We use fastText~\cite{Bojanowski2017} to obtain embeddings for tokens.
Unlike other word embedding models~\cite{Mikolov2013,Pennington2014} that learn a distinct embedding vector for each word, fastText learns embeddings for character n-grams and computes the embedding for a word as the sum of the embeddings of all n-grams appeared in that word.
As a result, fastText can naturally handle rare tokens, whereas other word embedding models often regard these tokens as out-of-vocabulary tokens and replace them with a special token such as ``UNK''.
We have emperically observed that fastText is more robust than alternative methods to common typos and misspelling, and can produce similar embeddings for homomorphically similar tokens. As a result, we choose fastText as our way to convert tokens to token embeddings.

\subsection{Attention-based Attribute Embedding}
\label{subsec:attribute-embedding}

The second step of \method takes the sequence of token embeddings for each attribute as input and outputs an embedding that encodes the information of that attribute.
This step is related to phrase/sentence embedding learning in NLP;\@ but the major challenges are that the nature of different attributes varies regarding their length, word choice, and usage, and that the sequential order of sentences in natural language is missing in tabular data.

We propose an attention-based attribute encoder (called \emph{attentional encoder} for short) for computing attribute embeddings. The main idea behind attentional encoders is \emph{averaging}---the embedding of an attribute is represented by a weighted average of its token embeddings.
%This choice is motivated by empirical observations that even simple unweighted averaging does well in representing short phrases~\cite{mikolov2013distributed} and by recent work that provides theoretical grounding for representing sentences as weighted averages of their token embeddings~\cite{Arora2016}.
But rather than fixing a weight for each token a priori, the attentional encoder learns the weight for each token depending on its semantics, position, and surrounding tokens in the input token sequence.
This capability is especially useful when attributes are long and exhibit clear structural patterns.
For example, the extra version descriptions in the song titles often (but not always) appear at the end of the titles and are enclosed by parenthesis or square brackets (see tuple 2 in Table~\ref{tab:noisy-pair}). Given enough positive pairs in the training data in which one member of the pair has such a version description but the other does not, the attentional encoder is able to recognize such patterns and pay less ``attention'' (i.e., assigning lower weights) to the tokens that form the version description at the end of music title.

Formally, let $\v_1, \v_2, \dots, \v_l \in \bbR^d$ be the sequence of token embeddings, where $l$ is the sequence length and $d$ is the dimension of token embeddings.
An attentional encoder computes the weights of the input tokens as follows:
\begin{align}
\h_1, \h_2, \ldots, \h_l &= \SeqEnc(\v_1, \ldots, \v_l), \label{eq:seq2seq}\\
\alpha_1, \alpha_2, \ldots, \alpha_l &= \mathrm{SoftMax}(\w^T \h_1, \ldots, \w^T \h_l), \label{eq:attention-weights} \\
\beta_k &= \rho \alpha_k + (1 - \rho) \frac{1}{l}, \quad \forall k \in [l].   \label{eq:final-weights}
\end{align}
Here, $\SeqEnc(\cdot)$ can be any neural network architecture that takes a sequential input and generates an output for every input position. Possible choices include the standard recurrent neural network (RNN), bidirectional long short-term memory network (Bi-LSTM), 1D convolutional neural network, and transformers~\cite{Vaswani2017}.
The hidden states are then transformed into \emph{attention weights} in~\eqref{eq:attention-weights}, which are further smoothed with uninformative weight $1/l$ in~\eqref{eq:final-weights}, controlled by a hyper-parameter $\rho \in [0, 1]$. Finally, the attribute embedding is defined as
\begin{equation}
g(\v_1, \ldots, \v_l) \defas \sum_{k = 1}^l \beta_k \v_k. \label{eq:attribute-emb}
\end{equation}

% Simple as they are, we choose such attentional encoders because they enjoy several notable advantages.
% First, the actual outputs of attentional encoders are the attention weights ${[\alpha_k]}_{k = 1}^l$, which lie on the probability simplex $\Delta^{l-1}$ with $l - 1$ degrees of freedom (defined as $\{\x \in \bbR^{l} | \x \geq \0, \1^T \x = 1\} $).
% Given that typically $l \ll d$, attentional encoders will have a more restricted output space than $\bbR^d$---the output space of models that directly compute the attribute embedding (e.g., an RNN that takes its last hidden state as the global output). This careful restriction of the output space can potentially benefit the training and reduce the chance of overfitting.
% Second, the attribute embedding output by an attentional encoder can be well interpreted through its weights, which is critical in practice as users often need to do error analysis and understand how false positives and false negatives occur.
% Finally, the attentional encoder degrades nicely to an unweighted average when $\rho$ is set to $0$. Doing so can be useful for short attributes such as person names, and can potentially reduce the chance of overfitting further.

\subsection{Tuple Signature}
\label{subsec:tuple-signature}

In the third step of \method, we would like to combine the attribute embeddings and generate representations at the tuple level, such that representations for matched tuples have large cosine similarity.
Before a deep dive into which model to use, let us first consider a more fundamental question: \emph{what would happen if we compress the information in a tuple into a single representation for blocking}?
%The following example provides some clues to the answer.

\begin{example}
  Consider three tuples for the same song with attributes on Title, Album, and Composer:
  \begin{equation*}
    \small
    \begin{aligned}
    \x_1 &= (\text{Me and Mrs.\ Jones}, \emptyset, \emptyset),  \\
    \x_2 &= (\text{Me and Mrs.\ Jones}, \text{Call Me Irresponsible}, \text{Michael Bubl\'e}), \\
    \x_3 & = (\text{Me \& Mrs.}, \text{Call Me Irresponsible}, \text{Michael Bubl\'e}),
    \end{aligned}
  \end{equation*}
  where $\emptyset$ denotes missing value.
  Intuitively, the embeddings need to be dominated by Title in order to ensure $emb(\x_1) \approx emb(\x_2)$.
  This would, however, imply that the similarity between $emb(\x_1)$ and $ emb(\x_3)$ is not high (as $\x_1$ and $\x_3$ differ on Title).
  \label{ex:paradox}
\end{example}

This example indicates that when tuples contain a wide variety of attributes and can possibly have many missing values, representing each tuple with only \emph{one} embedding vector would result in low similarity for certain positive pairs.
Consequently, one would have to lower the similarity threshold $\theta$ in order to retrieve pairs such as both $(\x_1, \x_2)$ and $(\x_2, \x_3)$.
A small $\theta$, however, would also incur many false positive pairs, making the P/E ratio unaffordably large.

To address this issue, we propose to generate \emph{multiple signatures} such that each signature only captures a partial, distinct aspect for tuples, and two tuples are considered similar (and thus regarded as a candidate pair for blocking) as long as they are similar for one signature. With this design, we are able to overcome the issue in Example~\ref{ex:paradox}, as shown in the next example.

\begin{example} Continue Example~\ref{ex:paradox}. Now suppose we have two signature functions  $sig_1(\cdot)$ and $sig_2(\cdot)$ applied on Title, and on Album and Composer, respectively. Then we will have $sig_1(\x_1) = sig_1(\x_2)$, $sig_2(\x_2) = sig_2(\x_3)$. Thus, regardless that $sig_1(\x_1) \not\approx sig_1(\x_3)$ and $sig_2(\x_1) \not\approx sig_2(\x_3)$, the two candidate pairs $(\x_1, \x_2)$ and $(\x_2, \x_3)$ can still be retrieved with large threshold $\alpha$ by $sig_1(\cdot)$ and $sig_2(\cdot)$, respectively.

  \label{ex:paradox-sovled}
\end{example}

Formally, let $\g_1, \ldots, \g_m \in \{\emptyset\} \cup \bbR^d $ denote the embeddings of the $m$ attributes of a tuple. We define the $s$-th signature function to be a weighted average over non-missing attributes, i.e.,
\begin{equation}
f^{(s)}(\g_1, \ldots, \g_m) \defas \sum_{j = 1}^m \bbI(\g_j \neq \emptyset) w_{sj}  \g_j,
\label{eq:signature-function}
\end{equation}
where $\w_s \defas {[w_{sj}]}_{j =1}^m \geq \0$ is a nonnegative weight to be estimated, $\bbI(\cdot)$ is the indicator function, and $f^{(s)}$ is set to be $\emptyset$ when $\bbI(\g_j \neq \emptyset) w_j$ is zero for all $j\in [m]$. Given $S$ such signature functions ${\{f^{(s)}(\cdot)\}}_{s = 1}^S$, and denoting the signature computed by $f^{(s)}(\cdot)$ for tuple $\x_i$ by $\f^{(s)}_i$, the final similarity metric used in \method is the maximum cosine similarity over $S$ pairs of signatures, namely
\begin{equation}
% \similarity(\x_i, \x_{i'}) \defas \max_{s = 1, \ldots, S} \frac{\langle \f^{(s)}_i, \f^{(s)}_{i'} \rangle }{\Vert \f^{(s)}_i \Vert_2 \cdot \Vert \f^{(s)}_{i'} \Vert_2 }.
\similarity(\x_i, \x_{i'}) \defas \max_{s = 1, \ldots, S} \cos \left(\f^{(s)}_i, \f^{(s)}_{i'} \right),
\label{eq:similarity-metric}
\end{equation}
where $\cos(\f, \f') \defas \langle \f, \f' \rangle / \left( \Vert \f \Vert_2 \cdot \Vert \f' \Vert_2 \right)$ is the cosine similarity, and we set the cosine similarity to zero if either $\f$ or $\f'$ is $\emptyset$.
Note that the cosine similarity is scale-invariant; we thus require $\Vert \w_s \Vert_2 = 1$ for all $s\in [S]$ without loss of generality.
\subsection{Model Training}
\label{subsec:supervised-training}

Now we describe how the proposed attentional encoders and signature functions are trained with the given positive label sets $\cL$. Our training algorithm is based on the following idea:
\emph{for any $(i, i') \in \cL$, the tuple pair $(\x_i, \x_{i'})$ should be more similar than pairs $(\x_i, \star)$ and $(\x_{i'},\star)$, where $\star$ denotes an irrelevant tuple}.

Embracing this intuition we design an auxiliary multi-class classification task for training.
Specifically, for each positive pair $(i, i) \in \cL$, we randomly sample a small set of indices $U_{i, i'} \subset [n] \setminus \{i, i'\}$.
Since typically $n \gg |U_{i, i'}|$ and duplicates in $\X$ are rare, one can reliably assume that the tuples corresponding to  $U_{i, i'}$ are irrelevant to $\x_i$ and $\x_{i'}$.
Then given signature function $f^{(s)}$, the probability of choosing $(\x_i, \x_{i'})$ to be the only positive pair among all $2|U_{i, i'}| + 1$ pairs from  $\{i, i'\} \cup U_{i, i'}$ that involve $\x_i$ or $\x_{i'}$  is defined as
\begin{equation}
\begin{aligned}
  & \Pr_s   \left[ (\x_i, \x_{i'}) ; U_{i, i'} \right] \\
  \defas & \frac{ e^{\similarity(\f^{(s)}_i, \f^{(s)}_{i'})} }
{  e^{\similarity(\f^{(s)}_i, \f^{(s)}_{i'})} + \sum_{j \in U_{i, i'}}  \big[ e^{\similarity(\f^{(s)}_{i}, \f^{(s)}_{j})} + e^{\similarity(\f^{(s)}_{i'}, \f^{(s)}_{j})} \big]  }.
\end{aligned}
\label{eq:selection-prob}
\end{equation}

The attentional encoders ${\{g^{(j)}\}}_{j = 1}^m$ and signature function weights $ {\{\w_s\}}_{s = 1}^S$ can thus be learned end-to-end by maximizing the summed log-probability over all signatures and all positive pairs, namely
\begin{equation}
\max_{ {\{g^{(j)}\}}_{j = 1}^m, {\{\w_s\}}_{s = 1}^S } \frac{1}{\vert \cL \vert }\sum_{(i, i') \in \cL} \sum_{s = 1}^S \log \Pr_s\left[ (\x_i, \x_{i'}) ; U_{i, i'} \right]
\label{eq:optim-prob}
\end{equation}
We apply Adam, a variant of stochastic gradient descent (SGD) algorithms, to optimize the objective. After each update, we further project all signature weights into the feasible region to ensure non-negativity and unit norm.

The optimization problem defined in~\eqref{eq:optim-prob}, however, does not impose any regularization on signature weights and thus may end up with $S$ identical, individually optimal signature functions.
%But we would like different signature functions to be uncorrelated and each to reflect a distinct aspect of tuples; consequently it is more likely to achieve high recall.
%A sufficient condition for the uncorrelation between signatures is orthogonality; namely, . To ensure such orthogonality, one could add it as an equality constraint to the optimization problem and apply, for example, penalty method that increasingly penalizes the original objective with the difference between $\W^T \W$ and $\I_S$, qualified by distance metrics for matrices such as Frobenius norm or log-determinant divergence.
%Yet penalty method is known to be ill-conditioned, as penalty coefficient needs to be large in order to enforce approximate feasibility~\cite{Nocedal2006}.
%We also find that other equality constrained optimization techniques such as augmented Lagrangian methods  , when applied together with SGD, requires a hard time tuning the hyperparameters similar to the penalty coefficient.
Ideally signatures should be independent, or orthogonal: $\W^T \W = \I_S$, so that each signature reflects a distinct aspect of tuples. Thus we could incorporate into the optimization problem a penalty such as $\left\Vert \W^T \W -  \I_S \right\Vert$, or use augmented Lagrangian methods~\cite{Nocedal2006}.
Nevertheless, when the optimization problem is situated into a larger task as here, tuning the penalty coefficient or related hyperparameters becomes unwieldy and impractical.

We instead propose a simple sequential algorithm to achieve orthogonality.
The main idea is that signature functions are trained one at a time, and when training the current signature function, all attributes used (i.e., associated with positive weights) by the previously identified signature functions are marked as unusable.
In this way, the attributes used by different signature functions will not overlap and thus satisfy orthogonality naturally.
Eventually, the algorithm terminates when either all attributes have been used, or $S$ signature functions have been learned.
In fact, this sequential algorithm not only eliminates the need for introducing new hyperparameters (as required by standard constrained optimization algorithms), it may also eliminate the need for tuning $S$: one could set the initial value of $\bar{S}$ as large as $m$ and let the sequential algorithm end up with an appropriate $S$.

Algorithm~\ref{alg:sequential-learning} sketches the final training procedure.

\begin{algorithm2e}[!tbp]
  \caption{Model Training}
  \label{alg:sequential-learning}

  \KwIn{tuple set $\X$, label set $\cL$,  maximum \# of iteration $T$, and maximum \# of signatures $\bar{S}$.}
  \KwOut{attribute encoders $ {\{g^{(j)}\}}_{j = 1}^m$ and signature function weights ${\{\w_s\}}_{s = 1}^S$}
  initialize $ {\{g^{(j)}\}}_{j = 1}^m$ and $ {\{\w_s\}}_{s = 1}^S$\;
  initialize attribute set $\cM := \{1, \ldots, m\}$\;
  \For{$s = 1, \ldots, \bar{S}$}{
    \For{$t = 1, \ldots, T$}{
      sample a mini-batch $\cL_B \subset \cL$\;
      sample a $U_{i, i'}$ for each $(i, i') \in \cL_B$\;
      update $ {\{g^{(j)}, w_{sj}\}}_{j \in \cM}$ to improve $ \frac{1}{\vert \cL_B \vert }\sum_{(i, i') \in \cL_B}  \log \Pr_s\left[ (\x_i, \x_{i'}) ; U_{i, i'} \right] $ \;
      project $\w_s$ to $\{\w \in \bbR_{\geq 0}^m | \w_{\cM} = 0, \Vert \w \Vert_2 = 1\}$\;
    }
    $\cM := \cM \setminus \{j \in [m] |\,w_{sj} > 0 \}$\;
    \If {$\cM = \emptyset $}{
      $S := s$; break\;
    }
  }
\end{algorithm2e}

\subsection{Fast NN Search}
\label{subsec:LSH}

In the last step of \method, our goal is to efficiently retrieve, for each query tuple, the nearest neighbors whose similarities to the query are above a specified threshold $\theta > 0$ according to the metric $\sigma$ defined by~\eqref{eq:similarity-metric} with the computed signatures. Note that $\sigma$ in~\eqref{eq:similarity-metric} takes a maximum form; thus we can conduct NN search for each signature function with threshold $\theta$ separately, then take the union of all candidate pairs found for each signature function (followed by a de-duplication step) as the final candidate pairs.

So the task is reduced to a classic, high-dimensional NN search problem with cosine similarity.
We choose to solve this problem with locality-sensitive hashing (LSH), an effective technique for this problem that offers \emph{provable} sublinear query time.
Specifically, we apply cross-polytope LSH, a state-of-the-art LSH family for cosine similarity that not only enjoys the theoretically optimal query time complexity but also allows an efficient implementation.

With cross-polytope LSH, one can prove that the NN search problem can be approximately solved in a sublinear query time, as shown in Theorem~\ref{thm:lsh-query-complexity}.
\begin{theorem}
  Given an $n$-point dataset $\X \subset \bbR^d$, there exists an algorithm based on cross-polytope LSH satisfying: for any query $\x$ and similarity thresholds $-1 < \theta' < \theta < 1$, if there exists a point $\x^* \in \X$ such that $ \cos(x, x^*) \geq \theta$,  the algorithm will with probability at least $1 - \varepsilon^K$  retrieve a point $\x' \in \X$ with $\cos(x, x') \geq \theta'$ in query time $\cO( K \cdot d \cdot n^{\rho}) $, where $\varepsilon < \frac{1}{3} + \frac{1}{e} $ and  $\rho = \frac{1 - \theta}{1 - \theta'} \cdot \frac{1 + \theta'}{1 + \theta} + o(1)$.
  \label{thm:lsh-query-complexity}
\end{theorem}
\begin{proof}
  The proof is in the Appendix.
  %~\footnote{Our supplementary is available at \url{hhttps://www.dropbox.com/s/1tb2efiux07hp7f}}.
\end{proof}

 Careful readers may have noticed that Theorem~\ref{thm:lsh-query-complexity} only guarantees the query time complexity for approximate instead of exact NN search. So why is a technique for approximate NN search sufficient in our case?
This is because by learning multiple signature functions, each of which focuses on only a particular aspect of a tuple, we empirically observe that resultant signatures are fairly similar for most positive pairs with similarities rarely falling below $0.8$. In contrast, the similarities of most random pairs center around $0.2$ and seldom exceed $0.4$. That means in our case, we can afford to set $\theta = 0.8$ and $\theta'= 0.4$ without incurring too many false positives. This would correspond to a sublinear query time complexity $\cO(K \cdot d \cdot n ^{1 / 3.86} )$.
The empirical evaluation in Section~\ref{subsec:scalability} further supports the effectiveness and scalability of the cross-polytope LSH.\@

\section{Empirical Evaluation}
\label{sec:experiments}
In this section, we empirically compare \method against an array of competitive baselines on three large-scale, real-world datasets.

\subsection{Experiment Setup}

\myparagraph{Datasets}
We consider three real-world datasets: \Movie, \Music, and \Grocery, crawled and sampled from various public websites (see Table~\ref{tab:datasets}).
\Music is from Amazon and Wikipedia, \Movie from IMDb and WikiData, and \Grocery from Amazon and ShopFoodEx (an online grocery store).
The three datasets represent three distinct dataset types in the entity matching problem:
\begin{enumerate}
  \item \emph{Clean} (\Movie):  Attributes are properly aligned and their values are relatively clean, i.e., each attribute rarely contains irrelevant information.
  \item \emph{Dirty} (\Music):  Certain attributes such as Title may contain significant amount of irrelevant information. In addition, attributes are imperfectly aligned and thus attribute values may be misplaced at the wrong attributes (see Composer and Song Writer in Table~\ref{tab:noisy-pair}).
  \item \emph{Unstructured} (\Grocery): Records are unstructured; that is, all information is mixed into a raw, relatively long, textual attribute.
\end{enumerate}

\begin{table}[htbp]
  \reducemargin
  \caption{Three real-world datasets for our experiments}
  \resizebox{\columnwidth}{!}{\begin{tabular}{llrrrr}
	\toprule
	\textbf{Dataset} & \textbf{Type} & \textbf{Table A} & \textbf{Table B} & \textbf{\# Pos.} & \textbf{\# Attr.} \\ \midrule
	\Movie            & Structured    &     465,893      &     202,162      &     135,275      &        8         \\
	\Music            & Dirty         &    2,190,080     &     105,446      &      2,298       &         7         \\
	\Grocery          & Unstructured       &    1,292,848     &      5,886       &      4,437       &         1         \\ \bottomrule
\end{tabular}}
  \label{tab:datasets}
  % \reducemargin
\end{table}

\myparagraph{Positive Label Generation}
For all three datasets we generate positive labels using the available strong keys. They are tconst (an alphanumeric unique identifier) for \Movie, ASIN (Amazon Standard Identification Number) for \Music, and  UPC code for \Grocery.
Note that not every record in the three datasets has such a strong key.
% Neither are these strong keys duplication-free; so for example the number of positive labels for \Grocery is larger than the size of Table B.
As these strong keys are used for constructing the positive labels for both training and evaluation, we exclude them from the attribute set to avoid overfitting.
\footnote{In practice, one can always add those pairs matched on these strong keys to the candidate blocking pairs.}

\myparagraph{Training/Test Set Split}
First, we randomly divide the positive labels into two parts, $80\%$ for training and $20\%$ for testing, and ensure that tuples sharing the same strong keys are put into the same part.
Then we create a training set that includes all tuples appearing in the training labels. We also add to the training set $20\%$ of the tuples that do not appear in any positive labels; they serve as irrelevant tuples to facilitate the training.
A test set is created in a similar manner, which includes all the tuples that appear in the test labels, as well as all remaining tuples.
We repeat this procedure and create five training/test set pairs for each dataset.

\myparagraph{Methods for Comparison}
We compare our method to a wide range of competitive baselines:

\begin{itemize}
  \item Key-based blocking: Two blocking choices are considered, i.e., single key (only Title is used as blocking key), disjunctive key (all attributes are used as blocking keys).

  %\item Customized key: first $n$ token, last $n$-suffixes, and token $n$-grams, with $n$ ranging from one to three.

  \item MinHash blocking: This method retrieves all pairs whose Jaccard similarities on a particular attribute are above $\theta$ as candidate pairs. All attributes are considered, and $\theta$ is set to be $0.4$, $0.6$, or $0.8$.

  \item DeepER~\cite{Ebraheem2018}: As a state-of-the-art, DL-based method for blocking, this method can be viewed as a special case of \method by setting all attribute encoders to unweighted averaging and letting each attribute be a signature.

\end{itemize}

To understand the impact of the components of \method, we also include two sub-model baselines:
\begin{itemize}
  \item Unweighted averaging encoder: It takes Title as the only signature and applies unweighted averaging to Title.

  \item Attentional encoder: It also takes Title as the only signature but applies the attentional encoder to Title.
\end{itemize}

\myparagraph{\method Configuration}
We use the pretrained fastText model with the word embedding size $d = 300$.
We choose the $\SeqEnc(\cdot)$ module in our attentional encoder to be a single-layer Bi-LSTM with $64$ hidden units, although we find our method is  robust to the choice of neural network architecture and other factors affecting optimization such as batch size and initial learning rate.
For each positive pair, we randomly sample  $|U_{i, i'}| \equiv 10$ irrelevant tuples on the fly during training to construct the loss defined in~\eqref{eq:selection-prob}.
We set the maximum number of signatures $S$ to be the number of attributes and let the Algorith~\ref{alg:sequential-learning} to determine the appropriate $S$.
We set the attention smoothing parameter $\rho$ to $1$ for attribute Title and to $0$ for other attributes.
%The number of identified signatures are $6$, $4$, and $1$ for \Movie, \Music, and \Grocery, respectively.
Finally, for NN search we set the similarity threshold $\theta = 0.8$ and  limit the maximum number of retrieved NNs for each tuple to $\max(1000, \lfloor \sqrt{n_1} \rfloor)$, where $n_1$ is the size of the larger table.

\myparagraph{Minimum Preprocessing}
For all methods, we only perform the same, minimum preprocessing to the datasets, as one of goals is to minimize human effort in blocking.
In fact, the only preprocessing we use is to convert English letters to lowercase and tokenize attributes into token sequences using the standard TreeBank tokenizer~\cite{Loper2009}.
All punctuation, stop-words, non-English characters, typos, abbreviations, and so forth are kept as they are.

\myparagraph{Evaluation Metrics}
We evaluate the effectiveness of each blocking method with two metrics--- recall and P/E ratio. Let $\cT \subseteq [n] \times [n]$ be the unknown set of all true matched pairs, and recall that $X$ is the tuple set and $\cC$ is the set of candidate pairs. The two metrics are defined as follows:
\begin{align*}
    recall & \defas \frac{|\cC \cap \cT|}{|\cT|}, \\
    P/E~ratio & \defas \frac{|\cC|}{|\X|} = \frac{|\cC|}{n}.
\end{align*}
The true label set $\cT$, however, is never known beforehand; we therefore approximate it with the collected positive labels $\cL$.
We report the average performances on the five test sets for each dataset.

\myparagraph{Additional Setups}
Due to the limit of space, we include additional experiment setup such as the implemention notes for our method and other baselines in the Appendix.

\subsection{Effectiveness}
\label{subsec:effectiveness}

\begin{table*}[ht]
  \centering
  \caption{Performance comparison of different methods on different datasets. The best and the second best recalls on each dataset are emboldened and italicized, respectively.  \method achieves the highest recall on dirty (\Music) and unstructured (\Grocery) datasets,  and a close second recall with a magnitude smaller P/E ratio on clean dataset (\Movie).}
  \setlength\doublerulesep{1pt}
\begin{tabular}{lrrrrrr}
	\toprule
	                                                                         & \multicolumn{2}{c}{\textbf{Movie}} & \multicolumn{2}{c}{\textbf{Music}} & \multicolumn{2}{c}{\textbf{Grocery}} \\
\cmidrule(lr){2-3}
\cmidrule(lr){4-5}
\cmidrule(lr){6-7}
\textbf{Method} &        Recall &          P/E ratio &        Recall &          P/E ratio &        Recall &            P/E ratio \\ \midrule
	Single Key (Title)                                                       &          58.6 &               0.10 &          72.3 &               0.40 &           0.0 &                 0.01 \\
	Disjunctive Key (All)                                                    &          85.5 &               9.13 &          92.5 &              39.11 &           0.0 &                 0.01
	% Token Blocking (All) &  	\\
	% Extended Suffix Array (All) & 91.0 	 & 111.8 	& 91.7 	& 261.97 	& 15.4 	& 0.08
	\\ \midrule
	MinHash (All, $\theta=0.8$)                                              &          86.1 &               9.24 &          89.5 &              38.52 &           0.8 &                 0.00 \\
	MinHash (All, $\theta=0.6$)                                              &          87.1 &              29.90 &          91.3 &              41.16 &          18.7 &                 0.84 \\
	%	MinHash (All, $\theta=0.5$)                                            &         0.933 &              62.29 &         0.915 &            1001.97 &         0.329 &                 0.98 \\
	MinHash (All, $\theta=0.4$)                                              & \textbf{91.4} &            2089.60 & \textit{96.2} &             232.98 &          52.5 &                 4.32 \\ \midrule
	%	MinHash (All, $\theta=0.3$)                                            &               &                    &         0.939 &            2688.49 &         0.838 &                68.65 \\ \midrule
	DeepER (All, $\theta=0.8$)                                               &          90.7 &             107.52 &          92.9 &              41.62 & \textit{70.8} &                 1.02 \\ \midrule \midrule
	Un. Avg. Enc. (Title, $\theta=0.8$)                                      &          62.5 &               7.65 &          77.5 &              11.62 &          70.8 &                 1.02 \\
	Atten. Enc. (Title, $\theta=0.8$)                                        &          64.0 &               4.74 &          94.4 &              21.00 &          89.1 &                 0.72 \\
  \bottomrule
	\method (All, $\theta=0.8$)                                              & \textit{90.8} &             105.06 & \textbf{96.3} &              48.80 & \textbf{89.1} &                 0.72 \\ \bottomrule
\end{tabular}

  \label{tab:recall_pe}
  % \reducemargin
\end{table*}

We begin with investigating the effectiveness of \method. Table~\ref{tab:recall_pe} summarizes the recall and P/E ratios of \method and other baseline methods on the three datasets. The results in the table support the following conclusions.

First, \method performs best overall, especially when datasets are dirty and/or unstructured.
On \Grocery, \method not only surpasses all baselines in recall by a substantial margin ($18.3$ percentage points, or $25.8\%$) but also attains the smallest P/E ratio.
On \Music, \method has higher recall than the leading baseline (MinHash with $\theta=0.4$), and its P/E ratio is only $1/5$ of the MinHash's P/E ratio.
On \Movie, \method achieves a close second recall, but its P/E ratio is about $20$ times smaller than the best baseline (MinHash with $\theta = 0.4$).

Second, key-based blocking fails to attain high recall on all three datasets.
This disadvantage is most evident on \Grocery, where not a single positive pair exactly matches on Title, because the two data sources (i.e., Amazon and ShopFoodEx) differ in their ways of concatenating different aspects of a grocery product (e.g., brand name, package size, and flavor) into a single attribute.
As a result, key-based methods are unable to retrieve any true positive pairs as candidate pairs (and thus have a recall of $0.0$).

Third, MinHash requires a low similarity threshold to achieve high recall but at a cost of unaffordably large P/E ratio.
In fact, only when $\theta=0.4$ MinHash achieves comparable recalls to \method on \Movie and \Music, but its P/E is substantially larger than \method's; when $\theta$ increases, MinHash's recall drops significantly.
This sensitivity to $\theta$ thus demands considerable amount of tuning in practice to achieve a balance between recall and P/E ratio.
Moreover, MinHash's recall is still much lower than \method on \Grocery even when $\theta=0.4$, suggesting that Jaccard similarity, which leverages only lexical evidence, is less effective than the similarity metric learned by \method on this challenging domain.

Fourth, the attention mechanism contributes significantly to \method's recall gain. This is best seen from the comparison between attentional encoder and unweighted averaging encoder, where the former outperforms the latter on \Music and \Grocery by $16.9$ and $18.3$ percentage points, respectively.
In addition, \method outperforms DeepER in recall on all three datasets, which further demonstrates the benefits of the attention mechanism.

Last but not least, learning multiple signatures further boosts the \method's recall. This is shown by the recall gain of \method over attentional encoder on \Movie and \Music.\footnote{Since \Grocery only has one attribute, \method is the same as attentional encoder and thus both methods share identical performance on this dataset. This convergence also happens to DeepER and the unweighted average encoder. }

\subsection{Automation}

We now explain how \method saves manual work but still obtains high recall.

One major source of the original manual work is concerned with how to  iteratively try out different combinations of cleaning and blocking key customization strategies.
This task is now alleviated by \method's ability to assign different weights to tokens through attentional encoders.
Figure~\ref{fig:attention-weights} visualizes the attention weights for the titles of sampled positive pairs on \Music (Figure~\ref{fig:attention-ex-1}) and \Grocery (Figure~\ref{fig:attention-ex-2}~\ref{fig:attention-ex-3}).
Several patterns stand out:

\begin{enumerate}
  \item The tokens at starting positions tend to enjoy large weights. This is consistent to our observation that positive pairs typically match on the first few tokens; these tokens may also encode important information such as brand in \Grocery (e.g., ``bertolli'' in Figure~\ref{fig:attention-ex-2} and ``la choy'' in Figure~\ref{fig:attention-ex-3}).

  \item Common stop words (e.g., ``the'' and ``and'' in Figure~\ref{fig:attention-ex-1}) and uninformative punctuation (e.g., most commas and periods) are properly ignored. This is also expected since these tokens are often irregularly injected into tuples and result in avoidable mismatches.

  \item The tokens in special positional relationship to  ``functional'' punctuation marks (e.g., parenthesis) tend to get special attention. For example, the  ``digitally'' and ``remastered'' in Figure~\ref{fig:attention-ex-1} are surrounded by parenthesis and get small weights.

  \item Many discriminative tokens (e.g., those regarding the package size, and flavor in \Grocery) are ignored. We were initially surprised at this because these tokens are usually useful in the downstream matching step. Later we realize that \method's choice may be reasonable because these tokens are often randomly missing or expressed in different forms; ignoring them in the blocking step avoids missing positive pairs.

\end{enumerate}

\begin{figure*}[!htbp]
  \reducemargin
  \centering
  % \hspace{-10mm}
  \subfloat[]{
    \raisebox{0.02in}{
      \includegraphics[width=0.30\textwidth]{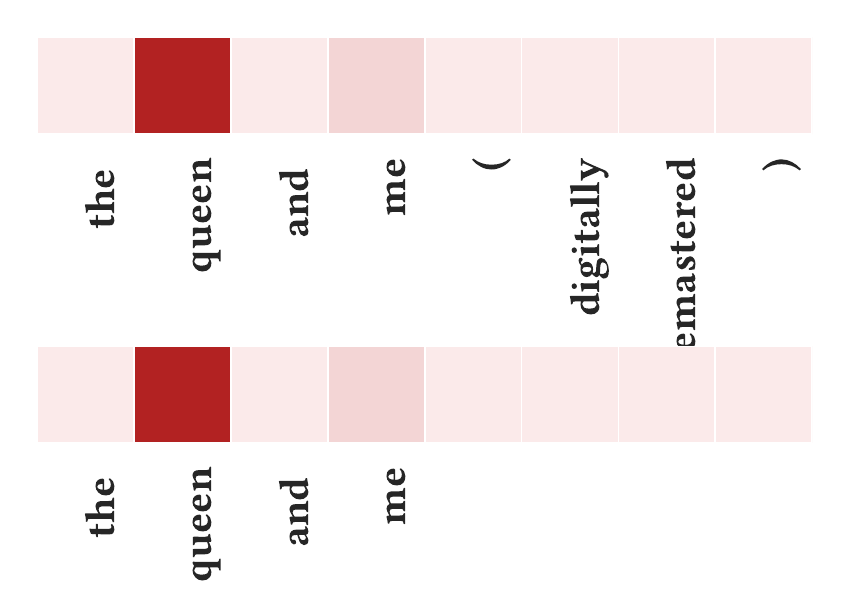}}
    \label{fig:attention-ex-1}
  }
  \subfloat[]{
    \includegraphics[width=0.30\textwidth]{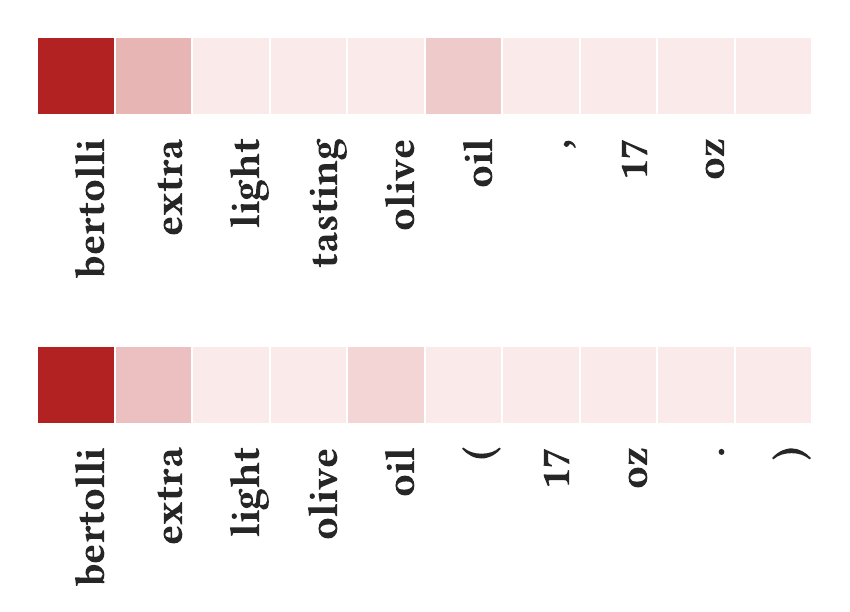}
    \label{fig:attention-ex-2}
  }
  \subfloat[]{
    \includegraphics[width=0.30\textwidth]{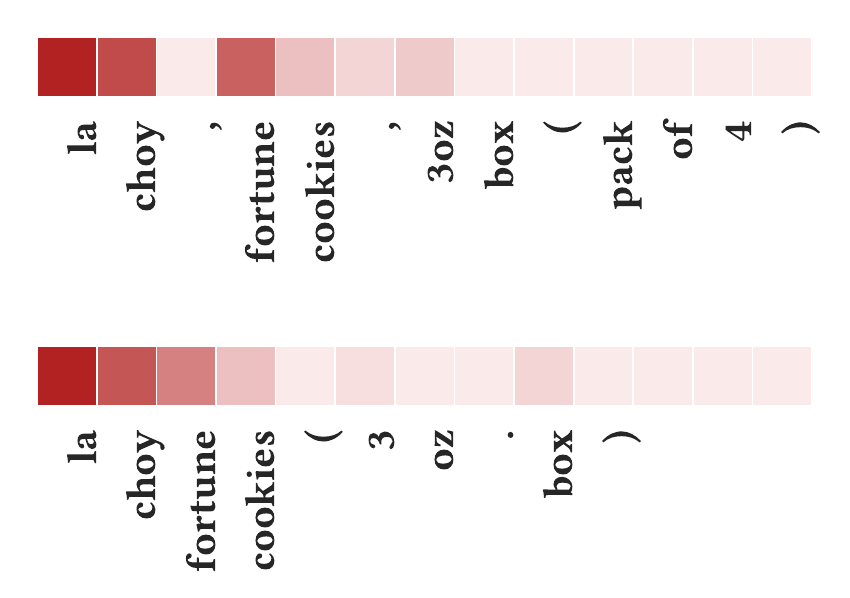}
    \includegraphics[width=0.04\textwidth]{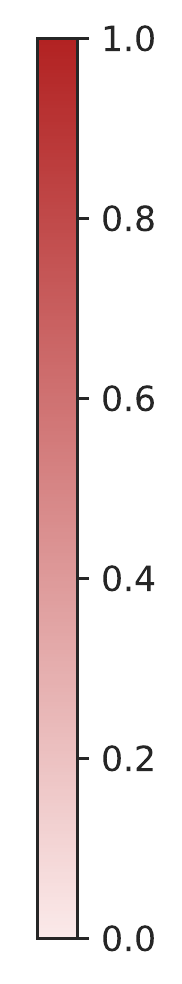}
    \label{fig:attention-ex-3}
  }
  \caption{Example attention weights on (a) \Music and (b, c) \Grocery. Important tokens are assigned to higher weights, representing by darker color.}
  \label{fig:attention-weights}
\end{figure*}

Manual work is also saved by \method's ability to automatically combine different attributes to generate signatures.
Figure~\ref{fig:signature-weights} shows the learned signature weights for \Music.
The combination of Ablums and Performers in $sig_2$ is likely due to the fact that they are often matched or unmatched simultaneously, and therefore combining them would reduce the chance of collision and reduce the P/E ratio.
The selection of Composer, Lyricist, and SongWriter by $sig_3$ allows an approximate cross-attribute matching among these three attributes, which is useful to handle the attribute misplacement cases in this dataset.

\begin{figure}[!htbp]
  \reducemargin
  \centering
  %   \null\hfill
  %     \subfloat[Movie]{
  %     \includegraphics[width=0.4\textwidth]{figures/sig_movie_A8.pdf}
  %   }
  %   \hfill
  %   \subfloat[Music]{
  %     \raisebox{0.2in}{
  %      \includegraphics[width=0.4\textwidth]{figures/sig_music.pdf}
  %   }}
  %   \hfill\null
  \includegraphics[width=0.8 \columnwidth]{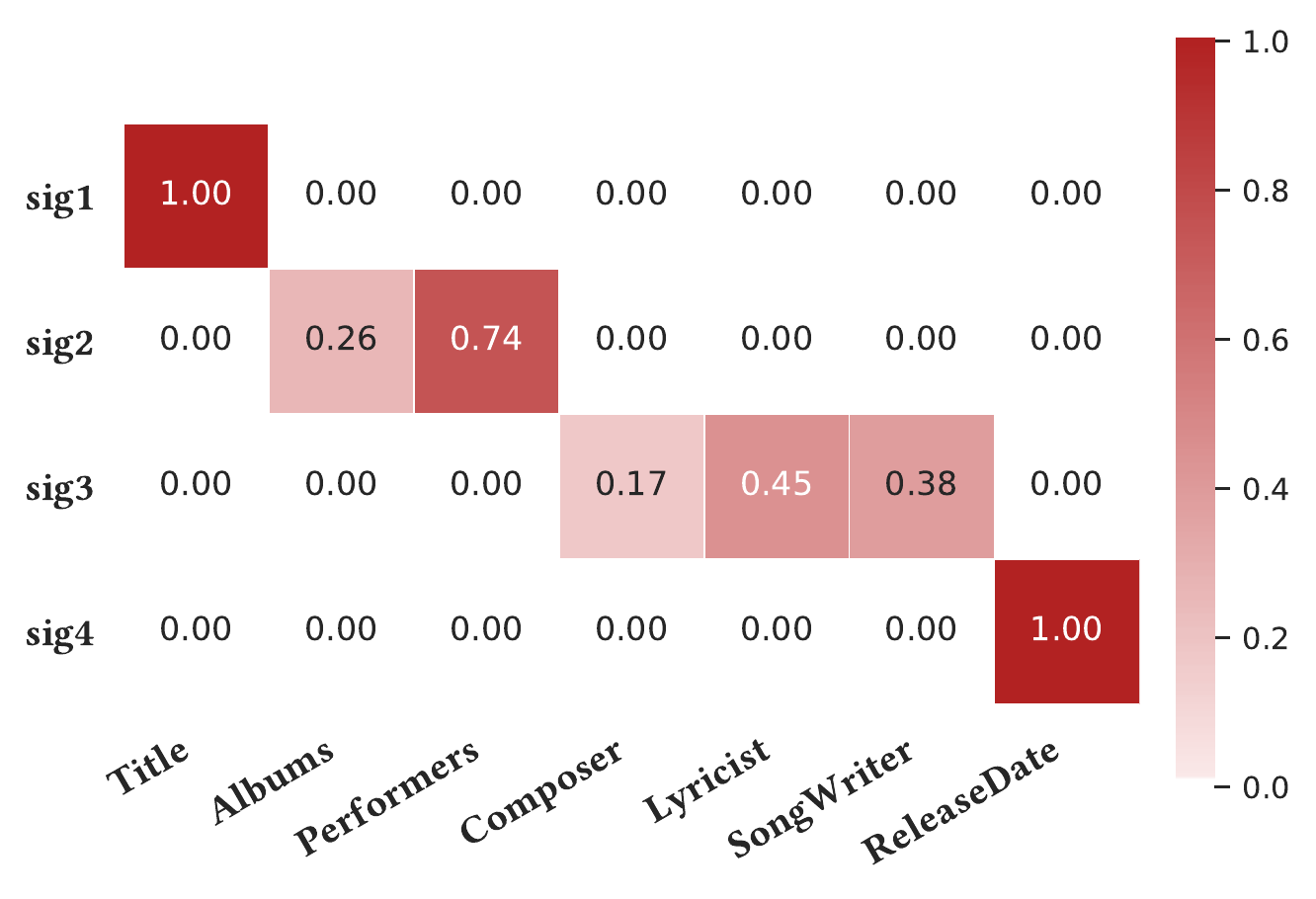}
  \caption{The learned signature weights on \Music. }
  \label{fig:signature-weights}
  \reducemargin
\end{figure}

\subsection{Scalability}

\label{subsec:scalability}

Finally, we investigate the empirical performance of cross-polytope LSH.\@
Figure~\ref{fig:avg-query-time-speedup} shows how much speedup of average query time (measured with single CPU core)  cross-polytope LSH can achieve over brute force with different the number of points in the retrieval set.
Huge speedup of cross-polytope over brute force is observed: when the number of points reaches $10^6$, the former is about 40--80 times faster than the latter.
In addition, the speedup improves sub-linearly as the number of points increase, which is expected since brute force has a strictly $\cO(n)$ query complexity, and so the speedup should scale $\cO(n^{1 - \rho})$ where $\rho < 1$.

%Figure~\ref{fig:avg-query-time-all}  how the average query time (measured with single CPU core) scales with respect to the number of points for cross-polytope LSH and for brute force. A huge gap between the two methods is observed: for cross-polytope LSH the curve grows slowly whereas brute force blows up quickly when the number of points reaches $10^5$.
%Then Figure~\ref{fig:avg-query-time-loglog} shows, in log-scale, a clear linear relationship between the average query time and the number points, which reinforces the sublinear query time complexity $\cO(n^\rho)$ given by Theorem~\ref{thm:lsh-query-complexity}.

%\begin{figure}[!tbp]
%  \reducemargin
%  \centering
%  \subfloat[LSH vs. Brute force]{
%    \includegraphics[width=0.5\columnwidth]{figures/avg_run_time_all.pdf}
%    \label{fig:avg-query-time-all}
%  }
%  \subfloat[LSH in Log-Scale]{
%    \includegraphics[width=0.5\columnwidth]{figures/avg_run_time_loglog.pdf}
%    \label{fig:avg-query-time-loglog}
%  }
%  \caption{The average query time of cross-polytope LSH on different datasets. (a) LSH is substantially faster than brute force. (b) The linear relationship between average time and the number of points in log-scale implies a sublinear query time complexity  $\cO(n^\rho)$.}
%  \label{fig:avg-query-time}
%  \reducemargin
%\end{figure}

\begin{figure}[!htbp]
  \reducemargin
  \centering
  \includegraphics[width=0.7\columnwidth]{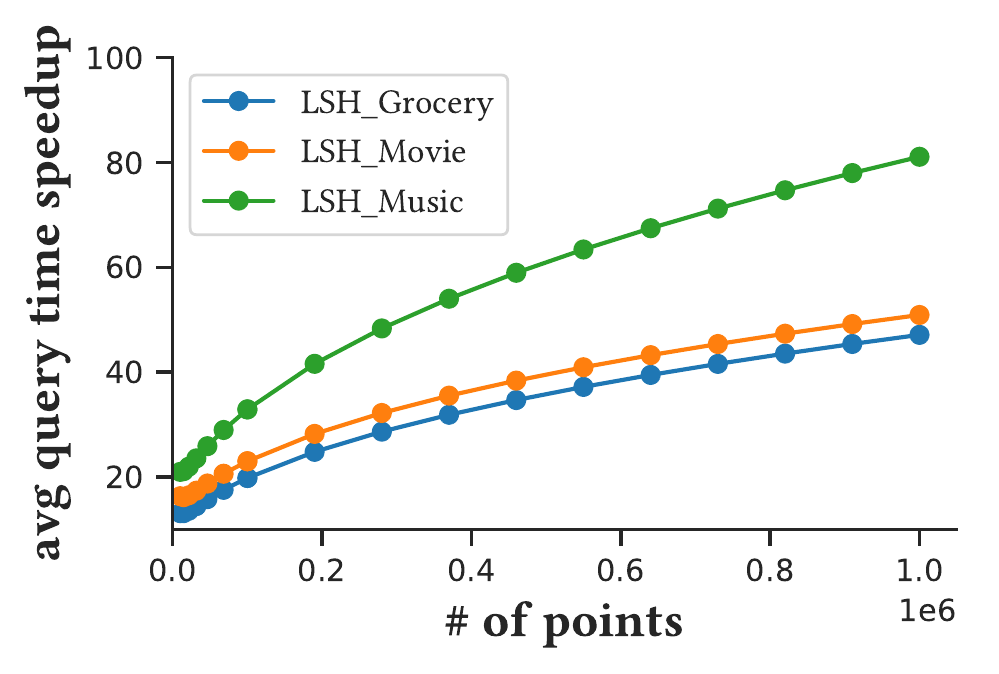}
  % \reducemargin

  \caption{The average query time speedups achieved by cross-polytope LSH over brute force on different datasets. LSH is substantially faster than brute force, and the speedup improves as the number of points increases.}
  \label{fig:avg-query-time-speedup}
  \reducemargin

\end{figure}

Theorem~\ref{thm:lsh-query-complexity} suggests a potential recall loss of cross-polytope LSH over exact NN search, as the algorithm may not retrieve any neighboring points with a small probability, or the retrieved points are all below the specified similarity threshold $\theta$. We therefore investigate how much recall loss this step of NN search using LSH can incur.
We experiment on \Grocery data, because its Table B has many fewer tuples than the Table A, which allows us to conduct exact NN search by brute force.
Table~\ref{tab:recall-gap-lsh} shows that under various similarity thresholds, although there is a recall gap between the two methods, the gap is very small---the largest gap (when $\theta=0.8$) is only $1.5\%$. We believe this is acceptable given the huge efficiency improvement of LSH upon brute force.
%Although slowly, the recall gap grows as the threshold $\theta$ decreases, suggesting the benefit of keeping similarities among positive pairs high enough.

\begin{table}[!tbp]
  \centering
  \caption{Comparison of recall between brute force NN search and LSH-based NN search on \Grocery. Only a minor recall gap of $0.89\%$ on average is observed.}

  \begin{tabular}{lrrrr}
\toprule 
& \multicolumn{4}{c}{\textbf{Threshold $\theta$}}\\ 
\cmidrule(lr){2-5}
    \textbf{Method} & 0.95 & 0.90 & 0.85 & 0.80  \\
\midrule 
    Brute force NN & 81.4 & 85.8 & 88.9 & 90.4 \\
    LSH-based NN          & 81.1 & 85.2 & 88.0 & 89.1 \\
% \midrule
%     Recall Gap   & 0.3 & 0.6 & 0.9 & 1.3 \\
\bottomrule     
\end{tabular}
  \label{tab:recall-gap-lsh}
  \reducemargin

\end{table}

% \begin{figure}
%   \centering
%   \includegraphics[width=0.8\columnwidth]{figures/recall_gap.pdf}
%   \caption{The recall gap between NN search with cross-polytope LSH and brutal force.}
%   \label{fig:recall-gap-LSH}
% \end{figure}

% As shown in Example~\ref{ex:paradox} and~\ref{ex:paradox-sovled}, given a wide variety of attributes, it is more achievable through learning multiple signatures rather than a single tuple-level embedding.

\section{Related Work}
\label{sec:related-work}
As a critical step of entity matching, blocking has been extensively studied over the last several decades with numerous methods being proposed.
For a comprehensive comparison of existing blocking methods, see~\cite{Papadakis2016}.
Among others, key-based blocking methods~\cite{Aizawa2005,Gravano2003,McNeill2012,Papadakis2013} are mostly used in practice yet require a lot of human effort.
MinHash blocking~\cite{Liang2014} allows a fuzzy match on attributes, but often ends up with unaffordably many candidate pairs.
The so-called ``meta-blocking''~\cite{Papadakis2014a,Papadakis2014,Papadakis2016a,Simonini2016} tries to reduce the P/E ratio by introducing---between blocking and matching---extra steps to prune the candidate pairs; their contributions are orthogonal to our work.
To the best of our knowledge, the recently proposed DeepER~\cite{Ebraheem2018} is the most relevant work to ours and can be viewed as a special case of \method.

Locality-sensitive hashing is firstly proposed in the seminal work~\cite{Indyk1998} for $\ell_p$ norm and later extended to other distance or similarity metrics. For cosine similarity, representative LSH schemes include hyperplane LSH~\cite{Charikar2002}, spherical LSH~\cite{Andoni2015}, and cross-polytope LSH~\cite{Andoni2015a}, among others.
While spherical LSH and cross-polytope LSH both attain the theoretically optimal query time complexity, only the latter can be efficiently implemented, as spherical LSH relies on rather complex hash functions that are very time-costly to evaluate.

\section{Conclusion}
\label{sec:conclusion}
We have proposed \method, a hands-off blocking framework for entity matching on tabular records, based on similarity-preserving representation learning and nearest neighbor search. 
Our contributions include: 
(a) \textbf{Automation}: \method frees users from tedious and laborious data cleaning and blocking key tuning. 
(b) \textbf{Scalability}: \method has a sub-quadratic total time complexity and can be easily deployed for millions of records. 
(c) \textbf{Effectiveness}: \method achieves superior performance on multiple large-scale, real-world datasets of various domains.
One future direction would be to extend \method to datasets with non-textual attributes (e.g., image, audio, and video).

\myparagraph{Acknowledgments}
We thank Xian Li, Tong Zhao, and Dongxu Zhang for the valuable discussions during the development of this work. We also thank the anonymous reviewers for their helpful feedback.

\bibliographystyle{ACM-Reference-Format}
\bibliography{reference}


\begin{thebibliography}{27}

%%% the Web address.

\ifx \showCODEN    \undefined \def \showCODEN     #1{\unskip}     \fi
\ifx \showDOI      \undefined \def \showDOI       #1{#1}\fi
\ifx \showISBNx    \undefined \def \showISBNx     #1{\unskip}     \fi
\ifx \showISBNxiii \undefined \def \showISBNxiii  #1{\unskip}     \fi
\ifx \showISSN     \undefined \def \showISSN      #1{\unskip}     \fi
\ifx \showLCCN     \undefined \def \showLCCN      #1{\unskip}     \fi
\ifx \shownote     \undefined \def \shownote      #1{#1}          \fi
\ifx \showarticletitle \undefined \def \showarticletitle #1{#1}   \fi
\ifx \showURL      \undefined \def \showURL       {\relax}        \fi
% invisible to TeX
\providecommand\bibfield[2]{#2}
\providecommand\bibinfo[2]{#2}
\providecommand\natexlab[1]{#1}
\providecommand\showeprint[2][]{arXiv:#2}

\bibitem[\protect\citeauthoryear{Aizawa and Oyama}{Aizawa and Oyama}{2005}]%
        {Aizawa2005}
\bibfield{author}{\bibinfo{person}{Akiko Aizawa} {and} \bibinfo{person}{Keizo
  Oyama}.} \bibinfo{year}{2005}\natexlab{}.
\newblock \showarticletitle{{A Fast Linkage Detection Scheme for Multi-Source
  Information Integration}}. In \bibinfo{booktitle}{\emph{International
  Workshop on Challenges in Web Information Retrieval and Integration}}.
  \bibinfo{publisher}{IEEE}, \bibinfo{pages}{30--39}.
\newblock
\showISBNx{0-7695-2414-1}
\urldef\tempurl%
\url{https://doi.org/10.1109/WIRI.2005.2}
\showDOI{\tempurl}


\bibitem[\protect\citeauthoryear{Andoni, Indyk, Laarhoven, Razenshteyn, and
  Schmidt}{Andoni et~al\mbox{.}}{2015}]%
        {Andoni2015a}
\bibfield{author}{\bibinfo{person}{Alexandr Andoni}, \bibinfo{person}{Piotr
  Indyk}, \bibinfo{person}{Thijs Laarhoven}, \bibinfo{person}{Ilya~P
  Razenshteyn}, {and} \bibinfo{person}{Ludwig Schmidt}.}
  \bibinfo{year}{2015}\natexlab{}.
\newblock \showarticletitle{{Practical and Optimal LSH for Angular Distance.}}.
  In \bibinfo{booktitle}{\emph{Advances in Neural Information Processing
  System}}. \bibinfo{pages}{1225--1233}.
\newblock
\showISSN{10495258}
\showeprint[arxiv]{1509.02897}
\urldef\tempurl%
\url{http://arxiv.org/abs/1509.02897}
\showURL{%
\tempurl}


\bibitem[\protect\citeauthoryear{Andoni and Razenshteyn}{Andoni and
  Razenshteyn}{2015}]%
        {Andoni2015}
\bibfield{author}{\bibinfo{person}{Alexandr Andoni} {and}
  \bibinfo{person}{Ilya~P Razenshteyn}.} \bibinfo{year}{2015}\natexlab{}.
\newblock \showarticletitle{{Optimal Data-Dependent Hashing for Approximate
  Near Neighbors.}}. In \bibinfo{booktitle}{\emph{STOC}}. ACM,
  \bibinfo{publisher}{ACM Press}, \bibinfo{address}{New York, New York, USA},
  \bibinfo{pages}{793--801}.
\newblock


\bibitem[\protect\citeauthoryear{Bilenko, Kamath, Mooney, View, and
  Mooney}{Bilenko et~al\mbox{.}}{2006}]%
        {Bilenko2006}
\bibfield{author}{\bibinfo{person}{Mikhail Bilenko}, \bibinfo{person}{Beena
  Kamath}, \bibinfo{person}{Raymond~J Mooney}, \bibinfo{person}{Mountain View},
  {and} \bibinfo{person}{Raymond~J Mooney}.} \bibinfo{year}{2006}\natexlab{}.
\newblock \showarticletitle{{Adaptive Blocking: Learning to Scale Up Record
  Linkage}}. In \bibinfo{booktitle}{\emph{Proceedings of IEEE International
  Conference on Data Mining}}. IEEE, \bibinfo{publisher}{IEEE},
  \bibinfo{pages}{87--96}.
\newblock
\showISBNx{0-7695-2701-7}
\showISSN{1550-4786}
\urldef\tempurl%
\url{https://doi.org/10.1109/ICDM.2006.13}
\showDOI{\tempurl}


\bibitem[\protect\citeauthoryear{Bojanowski, Grave, Joulin, and
  Mikolov}{Bojanowski et~al\mbox{.}}{2017}]%
        {Bojanowski2017}
\bibfield{author}{\bibinfo{person}{Piotr Bojanowski}, \bibinfo{person}{Edouard
  Grave}, \bibinfo{person}{Armand Joulin}, {and} \bibinfo{person}{Tomas
  Mikolov}.} \bibinfo{year}{2017}\natexlab{}.
\newblock \showarticletitle{{Enriching Word Vectors with Subword Information.}}
\newblock \bibinfo{journal}{\emph{Transactions of the Association for
  Computational Linguistics}}  \bibinfo{volume}{5} (\bibinfo{year}{2017}),
  \bibinfo{pages}{135--146}.
\newblock
\showISBNx{9781577357384}
\showISSN{10450823}
\urldef\tempurl%
\url{https://doi.org/1511.09249v1}
\showDOI{\tempurl}
\showeprint[arxiv]{1607.04606}


\bibitem[\protect\citeauthoryear{Charikar}{Charikar}{2002}]%
        {Charikar2002}
\bibfield{author}{\bibinfo{person}{Moses~S Charikar}.}
  \bibinfo{year}{2002}\natexlab{}.
\newblock \showarticletitle{{Similarity Estimation Techniques from Rounding
  Algorithms}}. In \bibinfo{booktitle}{\emph{Proceedings of the 34th Annual ACM
  Symposium on Theory of Computing}}, \bibfield{editor}{\bibinfo{person}{John~H
  Reif}} (Ed.). ACM, \bibinfo{publisher}{ACM}, \bibinfo{pages}{380--388}.
\newblock
\urldef\tempurl%
\url{https://doi.org/10.1145/509907.509965}
\showDOI{\tempurl}


\bibitem[\protect\citeauthoryear{Christen}{Christen}{2012}]%
        {Christen2012}
\bibfield{author}{\bibinfo{person}{Peter Christen}.}
  \bibinfo{year}{2012}\natexlab{}.
\newblock \showarticletitle{{A Survey of Indexing Techniques for Scalable
  Record Linkage and Deduplication}}.
\newblock \bibinfo{journal}{\emph{IEEE Transactions on Knowledge and Data
  Engineering}} \bibinfo{volume}{24}, \bibinfo{number}{9}
  (\bibinfo{year}{2012}), \bibinfo{pages}{1537--1555}.
\newblock
\showISBNx{1041-4347}
\showISSN{10414347}
\urldef\tempurl%
\url{https://doi.org/10.1109/TKDE.2011.127}
\showDOI{\tempurl}


\bibitem[\protect\citeauthoryear{Doan, Ardalan, Ballard, Das, Govind, Konda,
  Li, Mudgal, Paulson, C., and Zhang}{Doan et~al\mbox{.}}{2017}]%
        {Doan2017}
\bibfield{author}{\bibinfo{person}{AnHai Doan}, \bibinfo{person}{Adel Ardalan},
  \bibinfo{person}{Jeffrey~R Ballard}, \bibinfo{person}{Sanjib Das},
  \bibinfo{person}{Yash Govind}, \bibinfo{person}{Pradap Konda},
  \bibinfo{person}{Han Li}, \bibinfo{person}{Sidharth Mudgal},
  \bibinfo{person}{Erik Paulson}, \bibinfo{person}{Paul Suganthan~G C.}, {and}
  \bibinfo{person}{Haojun Zhang}.} \bibinfo{year}{2017}\natexlab{}.
\newblock \showarticletitle{{Human-in-the-Loop Challenges for Entity Matching:
  A Midterm Report}}. In \bibinfo{booktitle}{\emph{Proceedings of the 2nd
  Workshop on Human-In-the-Loop Data Analytics}},
  \bibfield{editor}{\bibinfo{person}{Carsten Binnig}, \bibinfo{person}{Joseph~M
  Hellerstein}, {and} \bibinfo{person}{Aditya~G Parameswaran}} (Eds.).
  \bibinfo{publisher}{ACM}, \bibinfo{address}{New York, New York, USA},
  \bibinfo{pages}{12:1----12:6}.
\newblock
\showISBNx{9781450350297}
\urldef\tempurl%
\url{https://doi.org/10.1145/3077257.3077268}
\showDOI{\tempurl}


\bibitem[\protect\citeauthoryear{Dong and Srivastava}{Dong and
  Srivastava}{2013}]%
        {Dong2013}
\bibfield{author}{\bibinfo{person}{Xin~Luna Dong} {and} \bibinfo{person}{Divesh
  Srivastava}.} \bibinfo{year}{2013}\natexlab{}.
\newblock \showarticletitle{{Big data integration}}.
\newblock \bibinfo{journal}{\emph{Proceedings of the VLDB Endowment}}
  (\bibinfo{year}{2013}), \bibinfo{pages}{1245--1248}.
\newblock
\showISBNx{978-1-4673-4910-9}
\showISSN{1063-6382}
\urldef\tempurl%
\url{https://doi.org/10.1109/ICDE.2013.6544914}
\showDOI{\tempurl}


\bibitem[\protect\citeauthoryear{Ebraheem, Thirumuruganathan, Joty, Ouzzani,
  and Tang}{Ebraheem et~al\mbox{.}}{2018}]%
        {Ebraheem2018}
\bibfield{author}{\bibinfo{person}{Muhammad Ebraheem},
  \bibinfo{person}{Saravanan Thirumuruganathan}, \bibinfo{person}{Shafiq Joty},
  \bibinfo{person}{Mourad Ouzzani}, {and} \bibinfo{person}{Nan Tang}.}
  \bibinfo{year}{2018}\natexlab{}.
\newblock \showarticletitle{{Distributed Representations of Tuples for Entity
  Resolution}}. In \bibinfo{booktitle}{\emph{Proceedings of the VLDB
  Endowment}}. \bibinfo{publisher}{VLDB Endowment},
  \bibinfo{pages}{1454--1467}.
\newblock
\showISSN{21508097}
\urldef\tempurl%
\url{https://doi.org/10.14778/3236187.3236198}
\showDOI{\tempurl}
\showeprint[arxiv]{1710.00597}


\bibitem[\protect\citeauthoryear{Gravano, Jagadish, Ipeirotis, Srivastava,
  Koudas, and Muthukrishnan}{Gravano et~al\mbox{.}}{2003}]%
        {Gravano2003}
\bibfield{author}{\bibinfo{person}{Luis Gravano}, \bibinfo{person}{H~V
  Jagadish}, \bibinfo{person}{Panagiotis~G Ipeirotis}, \bibinfo{person}{Divesh
  Srivastava}, \bibinfo{person}{Nick Koudas}, {and} \bibinfo{person}{S
  Muthukrishnan}.} \bibinfo{year}{2003}\natexlab{}.
\newblock \bibinfo{booktitle}{\emph{{Approximate String Joins in a Database
  (Almost) for Free--Erratum}}}.
\newblock \bibinfo{type}{{T}echnical {R}eport}.
  \bibinfo{institution}{Department of Computer Science, Columbia University}.
\newblock
\urldef\tempurl%
\url{https://doi.org/10.7916/D8M90HHN}
\showDOI{\tempurl}


\bibitem[\protect\citeauthoryear{Har-Peled, Indyk, and Motwani}{Har-Peled
  et~al\mbox{.}}{2012}]%
        {Har-Peled2012}
\bibfield{author}{\bibinfo{person}{Sariel Har-Peled}, \bibinfo{person}{Piotr
  Indyk}, {and} \bibinfo{person}{Rajeev Motwani}.}
  \bibinfo{year}{2012}\natexlab{}.
\newblock \showarticletitle{{Approximate Nearest Neighbor: Towards Removing the
  Curse of Dimensionality}}.
\newblock \bibinfo{journal}{\emph{Theory of Computing}} \bibinfo{volume}{8},
  \bibinfo{number}{1} (\bibinfo{date}{jul} \bibinfo{year}{2012}),
  \bibinfo{pages}{321--350}.
\newblock


\bibitem[\protect\citeauthoryear{Indyk and Motwani}{Indyk and Motwani}{1998}]%
        {Indyk1998}
\bibfield{author}{\bibinfo{person}{Piotr Indyk} {and} \bibinfo{person}{Rajeev
  Motwani}.} \bibinfo{year}{1998}\natexlab{}.
\newblock \showarticletitle{{Approximate nearest neighbors: towards removing
  the curse of dimensionality}}. In \bibinfo{booktitle}{\emph{Proceedings of
  the 30th annual ACM symposium on Theory of computing}}. ACM,
  \bibinfo{pages}{604--613}.
\newblock


\bibitem[\protect\citeauthoryear{Liang, Wang, Christen, and Gayler}{Liang
  et~al\mbox{.}}{2014}]%
        {Liang2014}
\bibfield{author}{\bibinfo{person}{Huizhi Liang}, \bibinfo{person}{Yanzhe
  Wang}, \bibinfo{person}{Peter Christen}, {and} \bibinfo{person}{Ross
  Gayler}.} \bibinfo{year}{2014}\natexlab{}.
\newblock \showarticletitle{{Noise-tolerant Approximate Blocking for Dynamic
  Real-time Entity Resolution}}. In \bibinfo{booktitle}{\emph{Pacific-Asia
  Conference on Knowledge Discovery and Data Mining}},
  Vol.~\bibinfo{volume}{8444}. Springer, \bibinfo{pages}{449--460}.
\newblock
\showISSN{16113349}
\urldef\tempurl%
\url{https://doi.org/10.1007/978-3-319-06605-9_37}
\showDOI{\tempurl}


\bibitem[\protect\citeauthoryear{Loper, Klein, and Bird}{Loper
  et~al\mbox{.}}{2009}]%
        {Loper2009}
\bibfield{author}{\bibinfo{person}{Edward Loper}, \bibinfo{person}{Ewan Klein},
  {and} \bibinfo{person}{Steven Bird}.} \bibinfo{year}{2009}\natexlab{}.
\newblock \bibinfo{booktitle}{\emph{{Natural Language Processing with
  Python}}}.
\newblock \bibinfo{publisher}{O'Reilly}.
\newblock


\bibitem[\protect\citeauthoryear{McNeill, Kardes, and Borthwick}{McNeill
  et~al\mbox{.}}{2012}]%
        {McNeill2012}
\bibfield{author}{\bibinfo{person}{Wiliam~P McNeill}, \bibinfo{person}{Hakan
  Kardes}, {and} \bibinfo{person}{Andrew Borthwick}.}
  \bibinfo{year}{2012}\natexlab{}.
\newblock \showarticletitle{{Dynamic Record Blocking: Efficient Linking of
  Massive Databases in MapReduce}}. In \bibinfo{booktitle}{\emph{Quality in
  Databases}}.
\newblock
\urldef\tempurl%
\url{http://www.purdue.edu/discoverypark/cyber/qdb2012/papers/1dynamicBlocking.pdf}
\showURL{%
\tempurl}


\bibitem[\protect\citeauthoryear{Michelson and Knoblock}{Michelson and
  Knoblock}{2006}]%
        {Michelson2006}
\bibfield{author}{\bibinfo{person}{Matthew Michelson} {and}
  \bibinfo{person}{Craig~a. Knoblock}.} \bibinfo{year}{2006}\natexlab{}.
\newblock \showarticletitle{{Learning Blocking Schemes for Record Linkage }}.
  In \bibinfo{booktitle}{\emph{Proceedings of the 21st national conference on
  Artificial intelligence}}. \bibinfo{pages}{2965--2967}.
\newblock
\showISBNx{1577352815}
\urldef\tempurl%
\url{https://doi.org/10.1039/b108224h}
\showDOI{\tempurl}


\bibitem[\protect\citeauthoryear{Mikolov, Corrado, Chen, and Dean}{Mikolov
  et~al\mbox{.}}{2013}]%
        {Mikolov2013}
\bibfield{author}{\bibinfo{person}{Tomas Mikolov}, \bibinfo{person}{Greg
  Corrado}, \bibinfo{person}{Kai Chen}, {and} \bibinfo{person}{Jeffrey Dean}.}
  \bibinfo{year}{2013}\natexlab{}.
\newblock \showarticletitle{{Efficient Estimation of Word Representations in
  Vector Space}}.
\newblock \bibinfo{journal}{\emph{arXiv preprint}}
  \bibinfo{volume}{2015-Janua}, \bibinfo{number}{3} (\bibinfo{year}{2013}),
  \bibinfo{pages}{1--12}.
\newblock
\showISBNx{1532-4435}
\showISSN{15324435}
\urldef\tempurl%
\url{https://doi.org/10.1162/153244303322533223}
\showDOI{\tempurl}
\showeprint[arxiv]{arXiv:1301.3781v3}


\bibitem[\protect\citeauthoryear{Nocedal and Wright}{Nocedal and
  Wright}{2006}]%
        {Nocedal2006}
\bibfield{author}{\bibinfo{person}{Jorge Nocedal} {and} \bibinfo{person}{Steven
  Wright}.} \bibinfo{year}{2006}\natexlab{}.
\newblock \bibinfo{booktitle}{\emph{{Numerical Optimization}}
  (\bibinfo{edition}{2} ed.)}.
\newblock \bibinfo{publisher}{Springer-Verlag New York}.
\newblock


\bibitem[\protect\citeauthoryear{Papadakis, Ioannou, Palpanas, Niederee, Nejdl,
  Nieder{\'{e}}, Nejdl, Niederee, and Nejdl}{Papadakis et~al\mbox{.}}{2013}]%
        {Papadakis2013}
\bibfield{author}{\bibinfo{person}{George Papadakis},
  \bibinfo{person}{Ekaterini Ioannou}, \bibinfo{person}{Themis Palpanas},
  \bibinfo{person}{Claudia Niederee}, \bibinfo{person}{Wolfgang Nejdl},
  \bibinfo{person}{Claudia Nieder{\'{e}}}, \bibinfo{person}{Wolfgang Nejdl},
  \bibinfo{person}{Claudia Niederee}, {and} \bibinfo{person}{Wolfgang Nejdl}.}
  \bibinfo{year}{2013}\natexlab{}.
\newblock \showarticletitle{{A Blocking Framework for Entity Resolution in
  Highly Heterogeneous Information Spaces}}.
\newblock \bibinfo{journal}{\emph{IEEE Transactions on Knowledge and Data
  Engineering}} \bibinfo{volume}{25}, \bibinfo{number}{12}
  (\bibinfo{year}{2013}), \bibinfo{pages}{2665--2682}.
\newblock
\urldef\tempurl%
\url{https://doi.org/10.1109/TKDE.2012.150}
\showDOI{\tempurl}


\bibitem[\protect\citeauthoryear{Papadakis, Koutrika, Palpanas, and
  Nejdl}{Papadakis et~al\mbox{.}}{2014a}]%
        {Papadakis2014a}
\bibfield{author}{\bibinfo{person}{George Papadakis}, \bibinfo{person}{Georgia
  Koutrika}, \bibinfo{person}{Themis Palpanas}, {and} \bibinfo{person}{Wolfgang
  Nejdl}.} \bibinfo{year}{2014}\natexlab{a}.
\newblock \showarticletitle{{Meta-blocking: Taking entity resolutionto the next
  level}}.
\newblock \bibinfo{journal}{\emph{IEEE Transactions on Knowledge and Data
  Engineering}} \bibinfo{volume}{26}, \bibinfo{number}{8}
  (\bibinfo{year}{2014}), \bibinfo{pages}{1946--1960}.
\newblock
\showISBNx{1041-4347 VO - 26}
\showISSN{10414347}
\urldef\tempurl%
\url{https://doi.org/10.1109/TKDE.2013.54}
\showDOI{\tempurl}


\bibitem[\protect\citeauthoryear{Papadakis, Papastefanatos, and
  Koutrika}{Papadakis et~al\mbox{.}}{2014b}]%
        {Papadakis2014}
\bibfield{author}{\bibinfo{person}{George Papadakis}, \bibinfo{person}{George
  Papastefanatos}, {and} \bibinfo{person}{Georgia Koutrika}.}
  \bibinfo{year}{2014}\natexlab{b}.
\newblock \showarticletitle{{Supervised meta-blocking}}.
\newblock \bibinfo{journal}{\emph{Proceedings of the VLDB Endowment}}
  \bibinfo{volume}{7}, \bibinfo{number}{14} (\bibinfo{date}{oct}
  \bibinfo{year}{2014}), \bibinfo{pages}{1929--1940}.
\newblock
\showISSN{21508097}
\urldef\tempurl%
\url{https://doi.org/10.14778/2733085.2733098}
\showDOI{\tempurl}


\bibitem[\protect\citeauthoryear{Papadakis, Papastefanatos, Palpanas, and
  Koubarakis}{Papadakis et~al\mbox{.}}{2016a}]%
        {Papadakis2016a}
\bibfield{author}{\bibinfo{person}{George Papadakis}, \bibinfo{person}{George
  Papastefanatos}, \bibinfo{person}{Themis Palpanas}, {and}
  \bibinfo{person}{Manolis Koubarakis}.} \bibinfo{year}{2016}\natexlab{a}.
\newblock \showarticletitle{{Boosting the Efficiency of Large-Scale Entity
  Resolution with Enhanced Meta-Blocking}}.
\newblock \bibinfo{journal}{\emph{Big Data Research}}  \bibinfo{volume}{6}
  (\bibinfo{year}{2016}), \bibinfo{pages}{43--63}.
\newblock
\showISSN{22145796}
\urldef\tempurl%
\url{https://doi.org/10.1016/j.bdr.2016.08.002}
\showDOI{\tempurl}


\bibitem[\protect\citeauthoryear{Papadakis, Svirsky, Gal, and
  Palpanas}{Papadakis et~al\mbox{.}}{2016b}]%
        {Papadakis2016}
\bibfield{author}{\bibinfo{person}{George Papadakis}, \bibinfo{person}{Jonathan
  Svirsky}, \bibinfo{person}{Avigdor Gal}, {and} \bibinfo{person}{Themis
  Palpanas}.} \bibinfo{year}{2016}\natexlab{b}.
\newblock \showarticletitle{{Comparative Analysis of Approximate Blocking
  Techniques for Entity Resolution}}. In \bibinfo{booktitle}{\emph{Proceedings
  of the VLDB Endowment}}, Vol.~\bibinfo{volume}{9}. \bibinfo{publisher}{VLDB
  Endowment}, \bibinfo{pages}{684--695}.
\newblock
\showISSN{21508097}
\urldef\tempurl%
\url{https://doi.org/10.14778/2947618.2947624}
\showDOI{\tempurl}


\bibitem[\protect\citeauthoryear{Pennington, Socher, and Manning}{Pennington
  et~al\mbox{.}}{2014}]%
        {Pennington2014}
\bibfield{author}{\bibinfo{person}{Jeffrey Pennington},
  \bibinfo{person}{Richard Socher}, {and} \bibinfo{person}{Christopher~D
  Manning}.} \bibinfo{year}{2014}\natexlab{}.
\newblock \showarticletitle{{Glove: Global Vectors for Word Representation.}}
\newblock \bibinfo{journal}{\emph{Proceedings of the 2014 Conference on
  Empirical Methods in Natural Language Processing (EMNLP)}}
  (\bibinfo{year}{2014}), \bibinfo{pages}{1532--1543}.
\newblock
\showISBNx{9781937284961}
\showISSN{10495258}
\urldef\tempurl%
\url{https://doi.org/10.3115/v1/D14-1162}
\showDOI{\tempurl}
\showeprint[arxiv]{1504.06654}


\bibitem[\protect\citeauthoryear{Simonini, Bergamaschi, and Jagadish}{Simonini
  et~al\mbox{.}}{2016}]%
        {Simonini2016}
\bibfield{author}{\bibinfo{person}{G Simonini}, \bibinfo{person}{S
  Bergamaschi}, {and} \bibinfo{person}{H~V Jagadish}.}
  \bibinfo{year}{2016}\natexlab{}.
\newblock \showarticletitle{{BLAST: A loosely schema-aware metablocking
  approach for entity resolution}}.
\newblock \bibinfo{journal}{\emph{Proceedings of the VLDB Endowment}}
  (\bibinfo{year}{2016}).
\newblock


\bibitem[\protect\citeauthoryear{Vaswani, Shazeer, Parmar, Uszkoreit, Jones,
  Gomez, Kaiser, and Polosukhin}{Vaswani et~al\mbox{.}}{2017}]%
        {Vaswani2017}
\bibfield{author}{\bibinfo{person}{Ashish Vaswani}, \bibinfo{person}{Noam
  Shazeer}, \bibinfo{person}{Niki Parmar}, \bibinfo{person}{Jakob Uszkoreit},
  \bibinfo{person}{Llion Jones}, \bibinfo{person}{Aidan~N Gomez},
  \bibinfo{person}{Lukasz Kaiser}, {and} \bibinfo{person}{Illia Polosukhin}.}
  \bibinfo{year}{2017}\natexlab{}.
\newblock \showarticletitle{{Attention is All you Need.}}. In
  \bibinfo{booktitle}{\emph{Advances in Neural Information Processing System}}.
\newblock
\showeprint[arxiv]{1706.03762}
\urldef\tempurl%
\url{http://arxiv.org/abs/1706.03762}
\showURL{%
\tempurl}


\end{thebibliography}

\clearpage

\appendix

% \onecolumn
% {
%   \centering \bf \huge
%   Supplementary Material for ``AutoBlock: A Hands-off Blocking Framework for Entity Matching''
%   \vspace{1em}
% }

% \begin{multicols*}{2}

% \begin{table}[h]
%  \caption{Notation table. Will be removed in the final draft. }
%  \input{tables/notation}
% \end{table}

% \begin{table}[htbp]
% \caption{\method matches all specs, while competitors
% miss one or more of the features.}
% \label{tab:salesman}
% \input{tables/salesman}
% \end{table}

\section{Additional Experiment Details}
\label{sec:appendix}

\subsection{Dataset}
The available attributes for each dataset are listed as follows:
\begin{itemize}
\item \Movie: Title, Description, Genres, Director, MusicComposer, Playwright, Characters, and Actors.
\item \Music: Title, Albums, Performers, Composer, Lyricist, SongWriter, and ReleaseDate.
\item \Grocery: Title.
\end{itemize}
For non-textual attributes (e.g., dates), we convert them into their text representations; for set-valued attributes (e.g., Actors), we concatenate the textual representation of all their elements.

\subsection{Implementation Details for \method}
We implement our model using \texttt{PyTorch}.  Different signatures share the same set of attribute encoders, but different attribute encoders have their own parameters.

In evaluating~\eqref{eq:selection-prob} it is possible that the signature $f^{(s)}$ may not be applicable to some tuples; that is, these tuples have missing values on all the attributes that correspond to the positive weights of $f^{(s)}$. If it happens to the sampled irrelevant tuples, i.e., for some $j \in U_{i, i'}$, we exclude $j$ from $U_{i, i'}$. But if the non-applicable tuple is either $\x_i$ or $\x_{i'}$, then we remove the positive pair $(i, i')$ from the mini-batch $\cL_B$. This encourages the learned signature functions to maximize the similarity gap between positive pairs and irrelevant pairs rather than to maximize the coverage of the signature functions.

The step of NN search with cross-polytope LSH is implemented with package \texttt{FALCONN}\footnote{\url{https://falconn-lib.org/}}. We build hash tables with the package's default configuration for $n=10^6$ points and dimension $p=300$, i.e., there are $K=10$ hash tables, and each table consists of $B=2$ hash function.
We multi-probe one additional bucket per table; that is, for each table, not only the points in the query's sitting bucket but also the points in the bucket that is closest to the query's sitting bucket are retrieved.

\subsection{Implementation Details for Baselines}
Our implementation of MinHash is from package \texttt{datasketch}.\footnote{\url{https://github.com/ekzhu/datasketch}}  There are $K=32$ hash tables and each table consists of $B=4$ hash functions.

\subsection{Platform and Total Runtime}
All experiments were conducted on a server with a 16-core CPU at 2.00Ghz and 128G memory. The total runtime for \method, from computing signatures to outputing final candidate pairs, is less than 0.5 hour for each dataset.

\section{Proof of Theorem~\ref{thm:lsh-query-complexity}}
\label{ap:proof-for-lsh-theorem}

\begin{proof}
Let $\cS^{p-1}$ be the unit sphere in $\bbR^p$, i.e., $\cS^{p - 1} \defas \{\x \in \bbR^p |\,\Vert x \Vert_2 = 1 \}$. Since cosine similarity is scale-invariant, we can project points onto $\cS^{p-1}$ without changing the cosine similarity among them. Hence, we can assume that $\X \subset \cS^{p-1}$ without loss of generality.

The Corollary 1 in~\cite{Andoni2015a}, together with Theorem 3.4 in~\cite{Har-Peled2012}, establishes that given an $n$-point dataset $\X \subset \cS^{p-1}$, there exists an algorithm based on hash tables built with cross-polytope LSH satisfying: for any query $\x$, Euclidean distance threshold $r > 0$, and approximation factor $c > 1$,
if there exists a point $\x^* \in \X$ such that $ \Vert x - x^* \Vert_2 < r$, the algorithm will with success probability at least $1 - \varepsilon$  retrieve a point $\x' \in \X$ with $\Vert x -  x' \Vert_2 < cr$ in query time $\cO(d \cdot n^{\rho})$, where $\varepsilon < \frac{1}{3} + \frac{1}{e} $ and
\begin{equation}
\rho = \frac{1}{c^2} \cdot \frac{4 - c^2 r^2}{4 - r^2}+ o(1).
\label{eq:rho_Euclidean}
\end{equation}
Thus, when $K$ independent copies of such hash tables are built, the success probability improves to at least $1 - \varepsilon^ K$, yet the query time complexity also increases to $\cO(K \cdot d \cdot n^{\rho})$.

Note that there is a one-to-one mapping between the cosine similarity and Euclidean distance for points on $\cS^{p - 1}$, i.e., $\cos(x, x') =  1 - \frac{1}{2} \Vert x - x' \Vert_2^2$. Plug $r = \sqrt{2 - 2 \theta}$ and $c = \sqrt{(1 - \theta') / (1 - \theta)}$ in~\eqref{eq:rho_Euclidean}, we get the result in Theorem~\ref{thm:lsh-query-complexity}.

\end{proof}

% \end{multicols*}

\end{document}